\input ./style/arxiv-general.cfg
\documentclass[MSNbibl,number,citesort,seceqn,dvips]{arxbj}
\makeatletter
   \@ifpackageloaded{graphicx}{}{\usepackage{graphicx}}
\makeatother
\usepackage{upgreek,dcolumn,multirow}
\usepackage{graphicx}

\aid{0}
\volume{21}
\issue{2}
\pubyear{2015}
\firstpage{957}
\lastpage{1001}
\doi{10.3150/13-BEJ594} % kopijuoti is 'New paper accepted'

\makeatletter

\newcommand{\rright}{\right}
\newcommand{\lleft}{\left}
\newcommand{\R}{\mathbb{R}}

\def\llambda{{\lambda\hspace{-5pt}\lambda}}

\newtheorem{theorem}{Theorem}[section]
\newtheorem{corollary}[theorem]{Corollary}
\newremark{remark}[theorem]{Remark}
\newtheorem{lemma}[theorem]{Lemma}

\def\th{\vartheta}
\def\eps{\varepsilon}

\def\e{\mathrm{e}}

\newcommand{\eqref}[1]{(\ref{#1})}

\renewcommand{\emptyset}{\varnothing}
\newcommand{\Cov}{\operatorname{Cov}}

\def\sfrac#1#2{#1/#2}
\def\vfrac#1#2{(#1)/#2}

\newcolumntype{d}[1]{D{.}{.}{#1}}
\newcommand{\bigtimes}{\mathop{\!\mbox{\parbox[c][9pt][b]{18pt}{\fontsize{18}{18}\selectfont{$\times$}}}\!\!\!\!}}
\makeatother

\begin{document}
\begin{frontmatter}

\title{Estimating failure probabilities}
\runtitle{Estimating failure probabilities}

\begin{aug}
%%%% inicialai - be tarpu
\author[A]{\inits{H.}\fnms{Holger} \snm{Drees}\corref{}\thanksref{A}\ead[label=e1]{holger.drees@math.uni-hamburg.de}} \and
\author[B]{\inits{L.}\fnms{Laurens} \snm{de Haan}\thanksref{B}\ead[label=e2]{ldehaan@ese.eur.nl}}
%%\runauthor{} %% auto
\address[A]{University of Hamburg, Department of Mathematics,
Bundesstr. 55, 20146 Hamburg, Germany.\\
\printead{e1}}
\address[B]{Erasmus University Rotterdam, Department of Economics,
P.O. Box 1738, 3000 DR Rotterdam, The Netherlands.
\printead{e2}}
\end{aug}

% HISTORY:
\received{\smonth{3} \syear{2012}}
\revised{\smonth{11} \syear{2013}}

% ABSTRACT
%
\begin{abstract}
In risk management, often the probability must be estimated that a
random vector falls
into an extreme failure set. In the framework of bivariate
extreme value theory, we construct an estimator for such failure
probabilities and analyze its asymptotic properties under natural
conditions. It turns out that the estimation error is mainly determined by
the accuracy of the statistical analysis of the marginal
distributions if the extreme value approximation to the dependence
structure is at least as accurate as the generalized Pareto
approximation to the marginal distributions. Moreover, we establish
confidence intervals and
briefly discuss generalizations to higher dimensions and issues
arising in practical applications as well.
\end{abstract}

% KEYWORDS
% visi is mazosios raides ir pagal abecele
%
\begin{keyword}
\kwd{asymptotic normality}
\kwd{exceedance probability}
\kwd{failure set}
\kwd{homogeneity}
\kwd{multivariate extremes}
\kwd{out of sample extrapolation}
\kwd{peaks over threshold}
\end{keyword}

\end{frontmatter}

%s1 #&#
\section{Introduction}\label{intro}

%s1.1 #&#
\subsection{Motivation}

Suppose an insurance company has contracts in two related lines of
business with all customers of an insurance portfolio (e.g., fire
insurance and business interruption insurance for industrial
customers). On top of quota reinsurances for both lines of business
(possibly with different quotas) the remaining total loss from each
incidence is covered by an excess of loss reinsurance (CAT-XL) that
pays for the part of the total loss which exceeds a given high
retention level $R$. If $X$ and $Y$ denote the original losses from
a fire in both lines of business and $1-\alpha_X$ and $1-\alpha_Y$
the corresponding quotas, then a claim occurs in the XL-reinsurance
if $\alpha_XX+\alpha_YY$ exceeds~$R$. For the purpose of risk
management, the reinsurer might be interested in the probability that
the insurance company will file a claim in case of a fire. If the
retention level is high, then the claim probability cannot be
estimated using simple empirical estimates, because in the past the
retention has rarely (or never) been exceeded.

In this paper, a more general setting is considered. We are
interested in estimating the probability that a pair of random
variables $(X,Y)$ will take on a value in some given ``extreme''
set. Similar problems arise naturally in many fields. For example, a
coastal dike may fail if the vector build from the still water level
and the wave heights lie in a certain failure set $D$ (cf. Coles and Tawn \cite{ct94}, Bruun and Tawn \cite{bt98}, and de
Haan and de Ronde \cite{hr98}).
A financial option (like a down-and-out-put) may become
worthless if the price vector of underlyings enters such a ``failure
set''. Finally, (part of) the principal of a catastrophe bond gets
lost for the investors if a vector of triggers becomes too extreme.

As there are insufficiently many observations available in the
extreme failure set $D$ to use standard statistical methods, extreme
value theory is needed to estimate the failure probability
$P\{(X,Y)\in D\}$.

%s1.2 #&#
\subsection{Extreme value approximations}
The basic idea of multivariate extreme value
theory is to assume that the suitably standardized componentwise
maxima of the observed random vectors converge to a non-degenerate
limit distribution. This assumption is
equivalent to the convergence of suitably standardized quantile
functions of all marginal distributions and a condition on the
dependence structure in extreme regions.

To be more precise, denote the marginal distribution functions of
$X$ and $Y$ by $F_1$ and $F_2$, respectively, and let $U_i(t):=
F_i^\leftarrow(1-1/t)$ with $H^\leftarrow(s):=\inf\{x\in\R\mid
H(x)\ge s\}$ denoting the generalized inverse of an increasing
function $H$. We assume that there exist real constants $\gamma_i$,
positive functions $a_i$ and real functions $b_i$ such that for
$x>0$ and $i\in\{1,2\}$
%
%e1.1 #&#
\begin{equation}
\label{eq:marginapprox} \lim_{t\to\infty} \frac{U_i(tx)-b_i(t)}{a_i(t)} =
\frac{x^{\gamma_i}-1}{\gamma_i}.
\end{equation}
For $\gamma_i=0$ read the right-hand side as $\log x$. Note that
the right-hand side is the $U$-function of the generalized Pareto
distribution (GPD) with distribution function
$1-(1+\gamma_ix)^{-1/\gamma_i}$ for $1+\gamma_ix>0$, that is to be
interpreted as the standard exponential distribution function for
$\gamma_i=0$. The parameter $\gamma_i$ is the so-called extreme
value index of the $i$th marginal. If it is positive, then the
support of $F_i$ is unbounded from above and $1-F_i(t)$ roughly
decays like the power function with exponent $1/\gamma_i$, while for
$\gamma_i<0$ the right endpoint $x_i^*:=F_i^\leftarrow(1)$ of the
support is finite and $1-F_i(x)$ roughly behaves like a multiple of
$(x_i^*-x)^{-1/\gamma_i}$ as $x\uparrow x_i^*$.

The aforementioned extremal dependence condition can be given in
terms of the standardized random variables $1-F_1(X)$ and
$1-F_2(Y)$, that are uniformly distributed on $[0,1]$ if the
marginal distributions are continuous. More precisely, we assume the
existence of a measure $\nu$ such that for $\nu$-continuous Borel
sets $B\subset[0,\infty)^2$ bounded away from the origin
%
%e1.2 #&#
\begin{equation}
\label{eq:nudef} \lim_{t\to\infty} t P\bigl\{(X,Y)\in U(tB)\bigr\} =
\nu(B).
\end{equation}
Here and in what follows, for functions $h_1$,
$h_2$ which are defined on subsets of the reals, we define a
function $h$ on a subset of $\R^2$ by
$h(x_1,x_2):=(h_1(x_1),h_2(x_2))$. The so-called exponent measure
$\nu$ describes the asymptotic dependence structure between extreme
observations $X$ and~$Y$. Its homogeneity property
%
%e1.3 #&#
\begin{equation}
\label{eq:homog} \nu(tB) = t^{-1} \nu(B),
\end{equation}
which holds for all Borel sets $B\subset[0,\infty)^2$ and all $t>0$,
will be pivotal for the construction of our estimator of the failure
probability. (Seen from a different angle, we assume an approximate
scaling law for the joint distribution of $U^\leftarrow(X,Y)$; cf.
Anderson \cite{a94}.) In addition, we need certain smoothness
assumptions
to ensure that $\nu$ does not have mass on the coordinate axes and not
too much mass in their neighborhoods (cf. condition (D2) in
Section~\ref{sec2.2}). Further details about the extreme value
assumptions can be found in de Haan and Ferreira \cite{hf06},
Sections~1.2 and 6.1, or Beirlant \textit{et al}. \cite{bgstwf04}, Chapters 2~and~8.\looseness=1

%s1.3 #&#
\subsection{Construction of estimators of extreme failure probabilities}
We are interested in the situation that at most a few observations
lie in the extreme failure set $D$ which implies that in our
mathematical framework the failure set $D=D_n$ must depend on the
sample size $n$ such that the failure probability
\[
p_n:=P\bigl\{(X,Y)\in D_n\bigr\}
\]
tends to 0. To motivate an estimator of $p_n$ based on independent
copies $(X_i,Y_i)$, $1\le i\le n$, of $(X,Y)$ first note that from
\eqref{eq:nudef} we obtain the approximation
%
%e1.4 #&#
\begin{equation}
\label{eq:nuapprox} \frac{n}k P \biggl\{ \frac{k}n
U^\leftarrow(X,Y) \in B \biggr\} \approx \nu(B)
\end{equation}
for any sequence $k=k_n\to\infty$ such that $k/n\to0$. To estimate
$p_n$ using this approximation, we must replace $U^\leftarrow$ and
$\nu$ with suitable estimators.

According to \eqref{eq:marginapprox}, we may approximate
$U_i((n/k)x)$ for sufficiently large $n$ by
%
%e1.5 #&#
\begin{equation}
\label{eq:Tnidef} T_{n,i}(x) := a_i(n/k)
\frac{x^{\gamma_i}-1}{\gamma_i}+b_i(n/k)
\end{equation}
and estimate it by
%
%e1.6 #&#
\begin{equation}
\label{eq:Tnihatdef} \hat T_{n,i}(x) := \hat a_i(n/k)
\frac{x^{\hat\gamma_i}-1}{\hat
\gamma_i}+\hat b_i(n/k),
\end{equation}
where $\hat a_i(n/k), \hat b_i(n/k)$ and $\hat\gamma_i$ are suitable
estimators for $a_i(n/k), b_i(n/k)$ and $\gamma_i$, respectively.
Likewise, the
generalized inverse functions $(k/n)U_i^\leftarrow(x)$ can be
estimated by
%
%e1.7 #&#
\begin{equation}
\label{eq:Tniinvhatdef} \hat T_{n,i}^\leftarrow(x) := \biggl(1+\hat
\gamma_i \frac{x-\hat b_i(n/k)}{\hat a_i(n/k)} \biggr)^{1/\hat\gamma_i}.
\end{equation}
Here and in the sequel, $(1+\gamma y)^{1/\gamma}$ is defined as
$\e^y$ if $\gamma=0$. For $1+\gamma y<0$ (or $1+\gamma y=0$ and
$\gamma<0$) the term $(1+\gamma y)^{1/\gamma}$ is not well defined.
If $\gamma$ is positive and $y<-1/\gamma$, then it may be
interpreted as 0, while for $\gamma<0$ and $y>-1/\gamma$ it may be
defined to be $\infty$. However, we will see that the precise
definition of $(1+\gamma y)^{1/\gamma}$ for very small and for
negative values of $1+\gamma y$ is not important in the present
setting (provided it is taken to be a non-decreasing function of
$y$), because the sets on which $\hat T_{n,i}^\leftarrow$,
$i\in\{1,2\}$, are not well defined are asymptotically negligible.

If, in \eqref{eq:nuapprox}, we substitute $\hat
T_n^\leftarrow(x_1,x_2):= (\hat T_{n,1}^\leftarrow(x_1), \hat
T_{n,2}^\leftarrow(x_2) )$ for the marginal transformation
$(k/n)U^\leftarrow$ and replace the probability in the left-hand
side of \eqref{eq:nuapprox} by its empirical counterpart, we arrive
at the following estimator of $\nu$
%
%e1.8 #&#
\begin{equation}
\label{eq:hatnundef} \hat\nu_n(B) := \frac{1}k \sum
_{i=1}^n \eps_{\hat
T_n^\leftarrow(X_i,Y_i)}(B),
\end{equation}
with $\eps_x$ denoting the Dirac measure with mass 1 at $x$.

Now, again interpreting convergence \eqref{eq:nudef} (for $t=e_n$)
as an approximation, we may estimate the failure probability as
follows:
%
%e1.9 #&#
%e1.10 #&#
\begin{eqnarray}
p_n & = & P\bigl\{(X,Y)\in D_n\bigr\}
\nonumber
\\
& = & P \bigl\{(X,Y)\in U\bigl(e_n\cdot e_n^{-1}U^\leftarrow(D_n)
\bigr) \bigr\}
\nonumber
\\
\label{eq:pnapprox1}& \approx& \frac{1}{e_n} \nu \bigl( e_n^{-1}
U^\leftarrow(D_n) \bigr)
\\
& \approx& \frac{1}{e_n} \nu \biggl( \frac{n}{ke_n} \hat
T_n^\leftarrow(D_n) \biggr)
\nonumber
\\
& \approx& \frac{1}{e_n} \hat\nu_n \biggl( \frac{n}{ke_n}
\hat T_n^\leftarrow(D_n) \biggr)
\nonumber
\\
\label{eq:pnhatdef}& =: & \hat p_n.
\end{eqnarray}

The basic idea of this estimator is to blow up the failure set,
after a standardization of the marginals, such that it contains
sufficiently many observations to allow the estimation of its
probability by an empirical probability. Note that, for given marginal
transformations, the estimator $\hat p_n$ depends on the tuning
parameters $k$ and $e_n$ only via their product $ke_n$, which controls
the factor by which the transformed failure set is blown up; see
Section~\ref{subsect:ken} for a detailed discussion. This factor
should be chosen by the statistician such that two contrary effects are
balanced. On the one hand,
$ke_n$ must not be too small, such that the inflated standardized
failure set $n/(ke_n) \hat T_n^\leftarrow(D_n)$ contains
sufficiently many marginally transformed observations $\hat
T_n^\leftarrow(X_i,Y_i)$, and thus the empirical probability
$\hat\nu_n ( n/(ke_n) \hat
T_n^\leftarrow(D_n) )$ is an accurate estimate of its expectation. On
the other hand, the set $e_n^{-1} U^\leftarrow(D_n)$ must be
sufficiently extreme to justify approximation \eqref{eq:pnapprox1}.
In Section~\ref{subsect:ken}, we discuss a heuristic tool to ensure
this balance.

%s1.4 #&#
\subsection{Alternative approaches} \label{subsect:alternatives}

An estimator related to $\hat p_n$ has been suggested and analyzed
by de Haan and Sinha \cite{hs99} in a much more restrictive
framework. In
particular, specific estimators for the marginal parameters have been
considered which use the same number $k_n$ of
largest order statistics for both marginal fits, which is
inefficient if the GPD approximation (cf. \eqref{eq:secordunif}
below) is less accurate for one of the marginal distributions.
Likewise, the flexibility of the estimator is increased in the
present paper by allowing that the blow-up factor $e_n$ deviates
from the unknown model constant $d_n$ defined below, while de Haan and Sinha \cite{hs99} used a consistent estimator of $d_n$ that was made
identifiable in a quite arbitrary way by fixing some point on the
boundary of some set $S$, which together with the factor $d_n$
determines the failure set $D_n$ (see \eqref{eq:Sdef}). In our
simulation study, it turns out that the inferior performance of the
estimator proposed by de Haan and Sinha is mainly caused by this often
inappropriate choice of $e_n$.

Moreover, the shape of the failure set considered by de Haan and
Sinha is restricted. For example, the case $q(\infty)=0$ (in our notation;
cf. condition (Q2) below) is ruled out by condition (2.9) of that
paper. The model assumption
\[
D_n:= \bigl\{(s,t)\mid f(s/x_n,t/y_n)\ge1
\bigr\}
\]
for some function $f$ and sequences of normalizing constants $x_n$
and $y_n$ seems quite restrictive and unnatural, because it allows
the failure set to tend towards the ``north-east'' only by a linear
scaling of both marginals. This parametrization does not fit well to
extreme value theory if the extreme value indices are not positive,
which is usually the case in environmetrics, one of the most important
fields of application of our theory besides financial risk management.

Even more troublesome
is the fact that by assumption (1.5) of de Haan and Sinha \cite{hs99} the
failure set is described in terms of the number $k_n$ of largest
order statistics that is picked by the statistician. Hence, the model
parametrization depends on the statistical procedure used to analyze
the model, which makes it extremely difficult to interpret.

Finally, while the influence of each marginal transformation is
clearly separated in the description of the limiting distribution in
our main Theorem~\ref{theo:main}, in Theorem~4.1 of de Haan and Sinha
\cite{hs99} the marginal parameters are seemingly intermingled.
Therefore, the generalization of the present results to higher
dimension is much more straightforward than those of de Haan and
Sinha (see the discussion in Section~\ref{subsect:multi}).

An alternative to our genuinely multivariate estimator can be
constructed by the so-called structural variable approach if the
failure set is of the form $D_n=\{(s,t)\mid h(s,t)\ge t_n\}$ for
some known function $h$ and threshold $t_n$. Then one may apply
techniques from univariate extreme value theory to the
pseudo-observations $h(X_i,Y_i)$, $1\le i\le n$ (cf. Coles
\cite{c01},
Chapter 8.2.4 and page 156, or Bruun and Tawn \cite{bt98}). However,
even for
this class of failure sets, an analysis of the dependence structure
between the two components of the observed vectors is of independent
interest, and it seems more natural to use the same approach for
model fitting and for the estimation of quantities like failure
probabilities. Moreover, often one wants to estimate the failure
probability for several different sets (e.g., to find the cheapest
construction to ensure a certain level of safety); in this case it
is both more efficient and more natural to use estimators in a
unified framework as considered in the present paper.

In the multivariate approach, Coles and Tawn \cite{ct94} and
Bruun and Tawn \cite{bt98} used parametric models for the
dependence structure in
the closely related problem to estimate a parameter defining a
failure set such that the corresponding failure probability equals a
given value. However, usually there is no physical reason for such
parametric models. By using them nevertheless, one trades a modeling
error, which is difficult to assess, for an estimation error, which
can be quantified at least asymptotically (see Theorem~\ref{theo:main} below). Having said this, it may be sensible to use a
parametric estimator of the failure probability if experience strongly
suggest that a simple model describes the data well. In that case, our
approach may be used as a countercheck of the model assumptions.

Note that our assumptions rule out that the exponent measure $\nu$
puts mass
on the coordinate axes. In particular, $X$ and $Y$ are assumed
asymptotically dependent in the sense of multivariate extreme value
theory in that $\lim_{t\to\infty} P(X>U_1(t)\mid Y>U_2(t))>0$. In
the case of asymptotic independent coordinates $X$ and
$Y$, consistency of an analogous estimator for the failure
probability was proved by Draisma \textit{et al}. \cite{ddfh04}, while its
asymptotic normality was established by M\"{u}ller \cite{m08}.

%s1.5 #&#
\subsection{Outline}

The paper is organized as follows: In Section~\ref{sect:main}, we
first introduce
and discuss in detail the framework in which we then prove
asymptotic normality of our estimator of the failure probability.
Moreover, we propose a consistent estimator of the limiting variance,
derive an asymptotic confidence interval, discuss the role of $ke_n$
and propose a
heuristic approach for choosing this factor. In Section~\ref
{sect:data}, we apply
the theory to the motivating example given at the beginning, while the
finite sample performance of the estimator is investigated in
Section~\ref{sect:simus}. All proofs are
collected in Section~\ref{sect:proofs}.

%s2 #&#
\section{Main results} \label{sect:main}

%s2.1 #&#
\subsection{Analysis of the estimation error}

The main goal of the present paper is to establish the asymptotic
normality of the estimator $\hat p_n$ under conditions on the
underlying distribution and the failure set which are easy to
interpret and relatively simple to verify. To achieve this objective,
we first decompose the estimation error into 6 parts. Loosely speaking,
the one that usually dominates the others (term $\mathit{IV}$ in equation \eqref
{eq:esterror} below) is due to the marginal fitting, two terms ($\mathit{II}$
and $\mathit{III}$) are related to the bias and the random error of the
estimator of the exponent measure, respectively, term $\mathit{VI}$ stems from
the approximation error in \eqref{eq:nuapprox}, while the remaining
two are related to a technical truncation argument.

To derive this decomposition,
recall that, in our asymptotic framework, the failure set $D_n$ must
become more extreme in the sense that it moves in the north-east
direction as the sample size $n$ increases to ensure that it
contains at most a few observations. To make both coordinates
comparable, we standardize the marginals using $U^\leftarrow$ and
assume that $U^\leftarrow(D_n)$ is essentially an increasing
multiple of a fixed set $S$. That way we ensure that none of the
coordinates dominates the other. More precisely, we assume that for
different sample sizes the failure sets are of the type
%
%e2.1 #&#
\begin{equation}
\label{eq:Sdef} D_n = U(d_nS) \cap\R^2 =
\bigl\{ \bigl(U_1(d_nx), U_2(d_ny)
\bigr)\mid(x,y)\in S\bigr\}\cap\R^2
\end{equation}
for a fixed
set $S\subset[0,\infty)^2$ and constants $d_n>0$
tending to $\infty$.
Note that from the analog to \eqref{eq:pnapprox1} where $e_n$ is
replaced with $d_n$ one obtains $d_n\approx\nu(S)/p_n$ (see Lemma~\ref{lemma:termVI} for a precise proof of the assertion $p_n d_n\to
\nu(S)$). Hence, the model constants $d_n$ determine at which rate
the failure probabilities tend to 0.

The crucial idea in the analysis of the asymptotic behavior of $\hat
p_n$ is to approximate the estimator by the empirical measure of a
\emph{random transformation} $H_n(S)$ of the set $S$ (with $H_n$
defined in \eqref{eq:Hndef} below) under the following analog to
$\hat\nu_n$ (defined in \eqref{eq:hatnundef}) with the fitted GPDs
replaced by the ``true'' ones:
\[
\nu_n(B) := \frac{1}k \sum_{i=1}^n
\eps_{
T_n^\leftarrow(X_i,Y_i)}(B).
\]
Since the GPD approximation of the marginals is accurate only in the
upper tail (and to avoid the aforementioned problem with the
definition of $T_n^\leftarrow$), we must first show that
asymptotically it does not matter if we replace $S$ with a suitably
defined subset $S_n^*$ that is bounded away from the coordinate
axes. For this set, we may use the approximation
\[
\hat p_n \approx\frac{1}{e_n}\nu_n \biggl(
\frac{d_n}{e_n}H_n\bigl(S_n^*\bigr) \biggr),
\]
where the random transformation $H_n$ of the marginals is defined
by
%
%e2.2 #&#
\begin{equation}
\label{eq:Hndef} H_n(x) := \frac{e_n}{d_n} T_n^\leftarrow
\circ\hat T_n \circ \bigl(\hat T_n^{(c)}
\bigr)^\leftarrow\circ U(d_nx)
\end{equation}
with
%
%e2.3 #&#
\begin{equation}
\label{eq:cndef} c = c_n := \frac{k}n e_n
\end{equation}
and
%
%e2.4 #&#
\begin{equation}
\label{eq:hatTncdef} \hat T_n^{(c)}(x,y)=\hat
T_n(c_nx,c_ny).
\end{equation}
Check that by \eqref{eq:marginapprox} one has $H_n(x)\approx
(e_n/d_n)(T_n^{(c)})^\leftarrow\circ U(d_nx)\approx
(e_n/d_n)(T_n^{(c)})^\leftarrow\circ T_n((k/\allowbreak n)d_nx)\approx x$ (cf.
Lemma~\ref{lemma:marginalapprox}).

Now, using the homogeneity of $\nu$,
we may break the estimation error into 6 parts as
follows:
%
%e2.5 #&#
\begin{eqnarray}\label{eq:esterror}
\hat p_n-p_n & = & \hat p_n -
\frac{1}{e_n} \nu_n \biggl(\frac{d_n}{e_n}H_n
\bigl(S_n^*\bigr) \biggr)
\nonumber
\\
& & { } + \frac{1}{e_n} \bigl( \nu_n(B)-E\nu_n(B)
\bigr)|_{B=(d_n/e_n)H_n(S_n^*)}
\nonumber
\\
& & { } + \frac{1}{e_n} \bigl( E\nu_n(B)-\nu(B)
\bigr)|_{B=(d_n/e_n)H_n(S_n^*)}
\nonumber
\\
& & { } + \frac{1}{d_n} \bigl( \nu\bigl(H_n
\bigl(S_n^*\bigr)\bigr) - \nu\bigl(S_n^*\bigr) \bigr)
\\
& & { } + \frac{1}{d_n} \bigl(\nu\bigl(S_n^*\bigr) - \nu(S)
\bigr)
\nonumber
\\
& & { } + \nu(d_n S)-p_n
\nonumber
\\
& =: & I+\mathit{II}+\mathit{III}+\mathit{IV}+V+\mathit{VI}.
\nonumber
\end{eqnarray}
It will turn out that, under suitable conditions, part $\mathit{IV}$ dominates
all the other terms. Its asymptotic behavior is largely determined
by the asymptotics of the marginal estimators if $\nu$ is
sufficiently smooth.

Under very weak conditions on the set $S$, we will show that the
terms $I$ and $V$ are negligible, if $S_n^*$ is defined suitably. If
$d_n/e_n$ is bounded and bounded away from 0, then using methods
from empirical process theory the second term can be shown to be
asymptotically negligible. Part $\mathit{VI}$ is a bias term which is
negligible if $d_n$ is sufficiently large (depending on the rate of
convergence in \eqref{eq:nudef}). Similarly, the term $\mathit{III}$, which
equals $ ((n/k)P\{T_n^\leftarrow(X,Y)\in\tilde B\}-\nu(\tilde
B) )/d_n$ for
$\tilde B=H_n(S_n^*)$, describes a bias term which is asymptotically
negligible if both the approximation \eqref{eq:nuapprox} and the
marginal approximation $U((n/k)B)\approx T_n(B)$ are sufficiently
accurate.

%s2.2 #&#
\subsection{Conditions for asymptotic normality}\label{sec2.2}

We will make the following assumptions about the marginal
distributions and the estimators of the marginal parameters:
\begin{enumerate}[(M3)]
\item[(M1)] There exist constants $x_i^0<F_i^\leftarrow(1)$ such that
$F_i$ is continuous and strictly increasing on
$[x_i^0,F_i^\leftarrow(1)]\cap\R$ for $i\in\{1,2\}$.

\item[(M2)] For all $i\in\{1,2\}$, there exist normalizing functions
$a_i>0$, $b_i\in\R$ and $A_i\ne0$ and constants $\rho_i<0$ such
that for all $x>0$
\[
\lim_{t\to\infty} \frac{%
\vfrac{U_i(tx)-b_i(t)}{a_i(t)} -
\vfrac{x^{\gamma_i}-1}{\gamma_i}}{A_i(t)} = \bar\psi_{\gamma_i,\rho_i}(x)
:= \lleft\{ %
\begin{array} {l@{\qquad}l} \displaystyle \frac{x^{\gamma_i+\rho_i}}{\gamma_i+\rho_i}, &\hspace*{-2.8pt}
\gamma_i+\rho_i\ne0,
\\
\log x, &\hspace*{-2.8pt} \gamma_i+\rho_i= 0. \end{array} %
\rright.
\]

\item[(M3)]
\[
k^{1/2} \biggl( \frac{\hat a_i(n/k)}{a_i(n/k)}-1,\frac{\hat
b_i(n/k)-b_i(n/k)}{a_i(n/k)},\hat
\gamma_i-\gamma_i \biggr)_{1\le i\le
2}
\longrightarrow(\alpha_i,\beta_i,\Gamma_i)_{1\le i\le
2}
\]
weakly.
\end{enumerate}

Condition (M1) is not crucial, but it is assumed to simplify the
proofs and the formulation of some technical results (cf. de
Haan and Ferreira \cite{hf06}, Theorem B.3.13).

(M2) is the usual second order condition with the additional
restriction that the second order parameters $\rho_i$ are negative.
Again, one may drop the latter assumption at the cost of additional
technical complications. According to Corollary~2.3.7 of de Haan and Ferreira \cite{hf06} we may and will assume that the normalizing
constants are chosen such that the following uniform version holds:
For all $\eps,\delta>0$ there exists $t_0$ such that
%
%e2.6 #&#
\begin{eqnarray}
\label{eq:secordunif} \biggl|\frac{%
\vfrac{U_i(tx)-b_i(t)}{a_i(t)} -
\vfrac{x^{\gamma_i}-1}{\gamma_i}}{A_i(t)} - \bar\psi_{\gamma_i,\rho_i}(x) \biggr| &\le&\delta
x^{\gamma_i+\rho_i}\max\bigl(x^\eps,x^{-\eps}\bigr)\nonumber\\[-8pt]\\[-8pt]
&=:& \delta
x^{\gamma_i+\rho_i\pm\eps}\nonumber
\end{eqnarray}
provided $t,tx>t_0$. In fact, the main results hold under the
following weaker assumption:
%
%e2.7 #&#
\begin{equation}
\label{eq:secordunif2} \biggl| \frac{U_i(tx)-b_i(t)}{a_i(t)} - \frac{x^{\gamma_i}-1}{\gamma_i} \biggr|=\mathrm{O}
\bigl(A_i(t) x^{\gamma_i+\rho_i\pm\eps} \bigr)
\end{equation}
as $t\to\infty$ uniformly for $x\ge t_0/t$. Under
condition (M2), $A_i$ is regularly varying with index $\rho_i$.

Condition (M3) gives a lower bound on the rate at which the
marginal estimators converge. Here some of the limiting
random variables may be equal to 0 almost surely. In particular,
this will usually be the case, if the $i$th marginal estimators use
$k_i$ largest order statistics and $k=\mathrm{o}(k_i)$. However, typically at least
some of the limiting random variables are non-degenerate and
jointly normally distributed. In the sequel, we
will choose versions such that the convergence in (M3) holds in
probability.

The failure set $D_n$ has to
satisfy the following conditions.
\begin{enumerate}[(Q2)]
\item[(Q1)] There exists a set
\[
S= \bigl\{(x,y)\subset[0,\infty)^2 \mid y\ge q(x)\ \forall x\in[0,
\infty) \bigr\} \subset[0,\infty)^2
\]
and constants $d_n>0$ tending to
$\infty$ such that
\[
D_n = U(d_nS)\cap\R^2 = \bigl\{
\bigl(U_1(d_nx), U_2(d_ny)\bigr)
\mid(x,y)\in S\bigr\} \cap\R^2.
\]
Here the function
$q\dvtx [0,\infty)\to[0,\infty]$, which describes the boundary of the
``archetypal failure set'' $S$,
is assumed
monotonically decreasing and continuous from the right with $q(0)>0$.

\item[(Q2)]
\begin{eqnarray*}
\begin{array}{rcl@{\qquad}l}
x^{(1-\gamma_1)/2}|\log x|&=&\mathrm{O}\bigl(q(x)\bigr) &\mbox{as } x\downarrow
x_l:=\inf\bigl\{ x\ge0\mid q(x)<\infty\bigr\},
\\\noalign{\vspace*{2pt}}
y^{(1-\gamma_2)/2}|\log y|&=&\mathrm{O}\bigl(q^\leftarrow(y)\bigr) &\mbox{as }
y\downarrow q(\infty):=\displaystyle \lim_{x\to\infty} q(x).
\end{array}
\end{eqnarray*}
\end{enumerate}
In particular, condition (Q1) ensures that one may define the
generalized (right-continuous) inverse function (of a decreasing
function) in the usual way:
\[
q^\leftarrow(v) := \inf\bigl\{ x>0\mid q(x)\le v\bigr\}
\]
with the convention $\inf\emptyset=\infty$. Roughly speaking, the
conditions (Q2) ensure that the archetypal failure set $S$ does not
have too much mass in a neighborhood of the axes where the estimated
marginal transformation often perform poorly. It is
always fulfilled if $\gamma_1\le1$ or $x_l>0$, resp., if
$\gamma_2\le1$ or $q(\infty)>0$.

Moreover, we need some conditions on the extremal dependence between
$X$ and $Y$ which is asymptotically
described by the exponent measure $\nu$ defined in \eqref{eq:nudef}.
In view of \eqref{eq:marginapprox}, one may replace the
standardization by $U$ with a standardization using $T_n$. To bound
the bias terms $\mathit{III}$ and $\mathit{VI}$ in \eqref{eq:esterror}, we must specify
the rate of the resulting convergence towards $\nu$:

(D1) There exist an exponent measure $\nu$ on $[0,\infty)^2$
and a function $A_0(t)>0$ converging to 0 as $t$ tends to $\infty$
such that
\[
t_n P \biggl\{ \biggl( \biggl(1+\gamma_1
\frac{X-b_1(t_n)}{a_1(t_n)} \biggr)^{1/\gamma_1}, \biggl(1+\gamma_2
\frac{Y-b_2(t_n)}{a_2(t_n)} \biggr)^{1/\gamma_2} \biggr)\in B \biggr\}-\nu(B)=\mathrm{O}
\bigl(A_0(t_n)\bigr)
\]
uniformly for all sets $B\in\mathcal{B}_{t_n,M}$ for $t_n=n/k$ and
for $t_n=d_n$ and arbitrary $M>0$.

Here, $\mathcal{B}_{t_n,M}$ consists of all sets of the form $ \{
(\tilde
H_{\th_i,\chi_i,\xi_i}^{(n,i)}(x_i))_{i\in\{1,2\}}\mid
(x_1,x_2)\in C \}$\vspace*{1.5pt} with $C=S\cap[u,\infty)\times[v,\infty)$ or
$C=[x_l,u)\times[q(u-),\infty)$ or
$C=[q^\leftarrow(v),\infty)\times[q(\infty),v)$ for some $u,v>0$
and some $\th_i,\chi_i,\xi_i\in[-M,M]$
if $t_n=n/k$, and $\mathcal{B}_{t_n,M}$ comprises all sets of the form
$ \{(
(1+\gamma_i(U_i(d_n x_i)-b_i(d_n))/a_i(d_n))^{1/\gamma_i})_{i\in\{
1,2\}}\mid
(x_1,x_2)\in C \}$ with $C=[x_l,u)\times[q(u-),\infty)$ or
$C=[q^\leftarrow(v),\infty)\times[q(\infty),v)$ for some $u,v>0$
if $t_n=d_n$. Here, for $i\in\{1,2\}$,
%
%e2.8 #&#
\begin{eqnarray}
\label{eq_tildeHnidef}  \tilde H_{\th_i,\chi_i,\xi_i}^{(n,i)}(x)
&:=& \biggl[ 1+\gamma_i \biggl( \frac{c_n^{-(\gamma_i-k^{-1/2}\th_i)}-1}{\gamma_i-k^{-1/2}\th
_i}
\bigl(1+k^{-1/2}\xi_i\bigr)\nonumber\\[-8pt]\\[-8pt]
&&\hphantom{ \biggl[}{} +c_n^{-(\gamma_i+k^{-1/2}\chi_i)}
\frac
{U_i(d_nx)-b_i(n/k)}{a_i(n/k)} \biggr) \biggr]^{1/\gamma_i}.
\nonumber
\end{eqnarray}

In (D1) the rectangles can also be replaced with the subsets $S\cap
((0,u)\times(0,\infty) )$, resp. $S\cap
((0,\infty)\times(0,v) )$. It is easy to see that condition
(D1) is met if $(X,Y)$ has a density $f$ and $\nu$ a density $\eta$ which
satisfy the following approximation
\begin{eqnarray*}
&&\sup_{(x,y)\in(0,\infty)^2,x\vee y\ge1} \frac{1}{w(x,y)} \biggl| t a_1(t)
a_2(t) x^{\gamma_1-1}y^{\gamma_2-1}
\\
& &\hphantom{\sup_{(x,y)\in(0,\infty)^2,x\vee y\ge1} \frac{1}{w(x,y)} \biggl|} { }\times f \biggl( a_1(t)\frac{x^{\gamma_1}-1}{\gamma_1}+b_1(t),a_2(t)
\frac{y^{\gamma
_2}-1}{\gamma_1}+b_2(t) \biggr)-\eta(x,y) \biggr| \\
&&\quad = \mathrm{O}
\bigl(A_0(t)\bigr)
\end{eqnarray*}
for some weight function $w$ which is Lebesgue-integrable on
$\{(x,y)\in(0,\infty)^2,x\vee y\ge1\}$. This sufficient condition
applies, for example, to the bivariate Cauchy distribution restricted to
$(0,\infty)^2$ and to densities of the form
$f(x,y)=1/(1+x^\alpha+y^\beta)$ with $\alpha,\beta>1$ such that
$1/\alpha+1/\beta<1$.

The
dependence may also be described by the pertaining spectral measure
$\Phi$ on $[0,\uppi/2]$ defined by
\[
\Phi\bigl([0,\th]\bigr)=\nu \biggl\{(x,y)\in[0,\infty)^2 \mid
x^2+y^2>1, \arctan\frac{y}x\le\th \biggr\},\qquad  \th
\in[0,\uppi/2].
\]

(D2) The spectral measure has a continuous
Lebesgue
density $\varphi$ on $[0,\uppi/2]$ such that
$\inf_{\delta\le t\le\uppi/2-\delta}\varphi(t)$ $>0$ for all
$\delta>0$ and
%
%e2.9 #&#
\begin{equation}
\label{eq:phivar} \lim_{\lambda\to1} \limsup_{t\downarrow0} \biggl|
\frac{\varphi
(\lambda t)}{\varphi(t)}-1 \biggr|+ \biggl|\frac{\varphi(\uppi/2-\lambda
t)}{\varphi(\uppi/2-t)}-1 \biggr|=0.
\end{equation}
This assumption rules out that the spectral measure (and hence the
exponent measure) puts mass on the coordinate axes. In particular, $X$ and
$Y$ must not be asymptotically independent (in the sense of multivariate
extreme value theory), because then the spectral measure is
concentrated on $\{0,\uppi/2\}$. Condition \eqref{eq:phivar} is
satisfied if $\varphi$ is
extended regularly varying at 0 and at $\uppi/2$ (cf. Bingham \textit{et al}.,
\cite{bgt87}, Section~2.0) or
if the function $\log\circ\,\varphi\circ\exp$ has a bounded derivative.

Condition (D2) is somewhat restrictive in that it requires the spectral
density to be bounded. Thus,
the exponent measure has a Lebesgue density $\eta$ given by
%
%e2.10 #&#
\begin{equation}
\label{eq:etanurel} \eta(x,y)=\bigl(x^2+y^2
\bigr)^{-3/2}\varphi \biggl(\arctan\frac{y}x \biggr),\qquad  x,y>0,
\end{equation}
which tends to 0 at the rate $(|x|+|y|)^{-3}$ as $|x|\to\infty$ or
$|y|\to\infty$.

Finally, we impose the following conditions on the sequences
$d_n,e_n$ and $k=k_n$:
\begin{enumerate}[(S3)]
\item[(S1)] $k\to\infty$, $n=\mathrm{O}(e_n)$ (so that $k=\mathrm{O}(c_n)$ with $c_n =
e_nk/n\to\infty$), $d_n\asymp e_n$ (i.e., $0<\liminf d_n/e_n\le
\limsup_{n\to\infty} d_n/e_n<\infty$),
and $w_n(\gamma_i)=\mathrm{o}(k^{1/2})$ for $i\in\{1,2\}$ with
\[
w_n(\gamma_i):= \lleft\{ %
\begin{array}
{l@{\qquad}l} \log c_n, & \gamma_i>0,
\\
 \frac{1}{2} \log^2 c_n, & \gamma_i=0,
\\\noalign{\vspace*{2pt}}
(d_nk/n )^{-\gamma_i}, & \gamma_i<0. \end{array}
\rright.
\]

\item[(S2)] $A_i(n/k)=\mathrm{o}(k^{-1/2} w_n(\gamma_i))$ for
$i\in\{1,2\}$ and $A_0(n/k)=\mathrm{o} (k^{-1/2} \max(w_n(\gamma
_1),\break w_n(\gamma_2)) )$

\item[(S3)]
 \begin{eqnarray*}
 \begin{array}{l@{\qquad}l}
k^{1/2}=\mathrm{O}\bigl(c_n\vee c_n^{\gamma_i}\bigr)&\mbox{if }\gamma_i\ge0\mbox{ for }i\in\{1,2\},\\\noalign{\vspace*{3pt}}
k^{1/2}=\mathrm{o}\bigl(c_n^{1-\gamma_1}\bigr)&\mbox{if }\gamma_1<0\mbox{ and }x_l=0,\mbox{ and}\\\noalign{\vspace*{3pt}}
k^{1/2}=\mathrm{o}\bigl(c_n^{1-\gamma_2}\bigr)&\mbox{if }\gamma_2<0\mbox{ and }q(\infty)=0.
\end{array}
\end{eqnarray*}
\end{enumerate}

Recall that $d_n$ is a constant determined by the model, which
describes the rate at which the failure probability $p_n$
tends to 0, while $e_n$ is chosen by the statistician such that the
inflated failure set contains sufficiently many observations. It
seems natural to choose $e_n$ of the same order as $d_n$, because
this way one compensates for the shrinkage of $D_n$. More precisely,
$d_n\asymp e_n$ if and only if the expected number of transformed
observations in the inflated transformed failure set is of the same
order as $k$, which can easily be checked in practical applications.
To see this, note that by \eqref{eq:nuapprox}, \eqref{eq:homog} and
\eqref{eq:Sdef} this expected number equals
%
%e2.11 #&#
\begin{eqnarray}
n P \biggl\{ \hat T_n^\leftarrow(X,Y)\in\frac{n}{ke_n} \hat
T_n^\leftarrow(D_n) \biggr\} & \approx& n P \biggl\{
\frac{k}n U^\leftarrow(X,Y) \in\frac{1}{e_n}U^\leftarrow(D_n)
\biggr\}
\nonumber
\\[-8pt]\\[-8pt]
& \approx& k \nu \biggl(\frac{d_n}{e_n} S \biggr) = k \frac{e_n}{d_n}
\nu(S).\nonumber
\end{eqnarray}
However, the condition $d_n\asymp e_n$ can be
substantially weakened at the price that one needs different
conditions for different combinations of signs of $\gamma_1$ and
$\gamma_2$.

The first condition of (S1) ensures that the expected number of
marginally standardized observations in the \emph{inflated}
standardized failure region tends to $\infty$, whereas the second
condition means that the expected number of observations in the
failure region remains bounded as $n\to\infty$. The last condition
of (S1) is needed to ensure consistency of the estimator in the
sense that $\hat p_n/p_n\to1$. It can only be satisfied
if $\min(\gamma_1,\gamma_2)>-1/2$. This restriction on the extreme
value indices usually arises if one wants to prove asymptotic
normality for estimators of tail probabilities; cf., for example, de
Haan and Ferreira \cite{hf06}, Remark~4.4.3, or Drees \textit{et al}. \cite{dhl06}, Remark~2.2.

From (S2) it follows that the bias is asymptotically negligible,
while (S3) will imply that the part of the set $S$ near the axes
(corresponding to observations where one of the coordinates is much
larger than the other) does not play an important role
asymptotically. Similarly as above, these conditions may also be
substantially weakened at the price of much more complicated
conditions on the behavior of $q$ depending on $\gamma_1,\gamma_2$
and $\eta$.

%s2.3 #&#
\subsection{Asymptotic approximation of
the estimator \texorpdfstring{$\hat p_n$}{$hat p_n$}}

Under the above condition, we establish the following approximation to
the estimation error of $\hat p_n$ in terms of the limiting random
variables of the marginal estimators.

%
%th2.1 #&#
\begin{theorem} \label{theo:main}
If the conditions \textup{(M1)--(M3), (D1), (D2), (Q1), (Q2)} and \textup{(S1)--(S3)} are
fulfilled,
then
%
%e2.12 #&#
\begin{eqnarray}\label{eq:mainapprox}
\hspace*{-25pt}&&k^{1/2}d_n(\hat p_n-p_n)
\nonumber
\\
&&\quad  =  w_n(\gamma_1)\lleft\{ %
\begin{array}{l@{\qquad  }l} - \displaystyle \frac{\Gamma_1}{\gamma_1} \displaystyle \int_{q(\infty)}^\infty
q^\leftarrow (v)\eta\bigl(q^\leftarrow(v),v\bigr) \,\mathrm{d}v, &
\gamma_1>0,
\\
\biggl(\displaystyle \frac{\alpha_1}{\gamma_1}-\beta_1-\displaystyle \frac{\Gamma_1}{\gamma
_1^2} \biggr) \displaystyle \int
_{q(\infty)}^\infty \bigl(q^\leftarrow(v)
\bigr)^{1-\gamma_1} \eta\bigl(q^\leftarrow(v),v\bigr)\, \mathrm{d}v, & \gamma
_1<0,
\\
- \Gamma_1 \displaystyle \int_{q(\infty)}^\infty
q^\leftarrow(v)\eta \bigl(q^\leftarrow(v),v\bigr) \,\mathrm{d}v, &
\gamma_1=0, \end{array} %
\rright.
\nonumber
\\[-8pt]\\[-8pt]
& & \qquad { } + w_n(\gamma_2) \lleft\{ %
\begin{array} {l@{\qquad  }l} - \displaystyle \frac{\Gamma_2}{\gamma_2} \displaystyle \int_{x_l}^\infty
q(u)\eta\bigl(u,q(u)\bigr) \,\mathrm{d}u, & \gamma_2>0,
\\
\biggl(\displaystyle \frac{\alpha_2}{\gamma_2}-\beta_2-\displaystyle \frac{\Gamma_2}{\gamma
_2^2} \biggr) \displaystyle \int
_{x_l}^\infty \bigl(q(u)\bigr)^{1-\gamma_2} \eta
\bigl(u,q(u)\bigr) \,\mathrm{d}u, & \gamma_2<0,
\\
- \Gamma_2\displaystyle \int_{x_l}^\infty q(u)\eta
\bigl(u,q(u)\bigr) \,\mathrm{d}u, & \gamma_2=0, \end{array} %
\rright.
\nonumber
\\
& &\qquad  { } + \mathrm{o}_P \bigl(w_n(\gamma_1)
\vee w_n(\gamma_2) \bigr).\nonumber
\end{eqnarray}
\end{theorem}

Since $p_nd_n\to\nu(S)$, Theorem~\ref{theo:main} remains true when
the left-hand side of \eqref{eq:mainapprox} is replaced with
$k^{1/2}\nu(S)(\hat p_n/p_n-1)$.

The weights $w_n(\gamma_1)$ and
$w_n(\gamma_2)$ on the right-hand side of \eqref{eq:mainapprox} may
be different, and then they converge to $\infty$ at different rates.
More precisely, $w_n(\gamma)$ is a non-increasing function of
$\gamma$, and it is strictly decreasing on $(-\infty,0]$. Therefore,
the smaller of both marginal extreme value indices $\gamma_1$ and
$\gamma_2$ determines the rate of convergence of $\hat p_n$ towards
$p_n$. If at least one of the indices is non-positive and the
indices are not equal, then the summand corresponding to the larger
index is negligible. (In that case, it may happen that one cannot
prove asymptotic normality using Theorem~\ref{theo:main}, because
the limiting random variables $\alpha_i,\beta_i$ and $\Gamma_i$
pertaining to the smaller extreme value index are equal to 0; cf. the
above discussion of condition (M3).)

If both extreme value indices are positive, then both main terms on
the right-hand side of \eqref{eq:mainapprox} are of the same order.
In that case, $(k^{1/2}d_n/\log c_n)(\hat p_n-p_n)$ converge to a
limit distribution which typically will be non-degenerate if at
least one of the limiting random variables $\Gamma_1$ and $\Gamma_2$
in (M3) is non-degenerate. If they are jointly normal, then we may
derive the asymptotic normality of the estimator for the failure
probability~$p_n$.

However, if the fit of the marginal tails by GPDs is much more accurate
than the approximation of the dependence structure by the extreme value
dependence structure described by the exponent measure, then all
limiting random variables in condition (M3) may be equal to 0, because
the marginal estimators are based on the largest $k_i\gg k$ order
statistics and converge at the rate $k_i^{-1/2}=\mathrm{o}(k^{-1/2})$. In that
case, Theorem~\ref{theo:main} merely specifies an upper bound on the
estimation error but not which of the terms $\mathit{I}$--$\mathit{VI}$ dominates the others.

%s2.4 #&#
\subsection{Asymptotic confidence intervals} \label{subsect:confint}

Theorem~\ref{theo:main} can be used to construct asymptotic
confidence intervals. To this end, it is advisable to reformulate
the assertion as a convergence result on $k^{1/2} e_n(\hat
p_n-p_n)$, because $d_n$ is unknown. Then one needs consistent
estimators for the variance of the random variables occurring on the
right-hand side of \eqref{eq:mainapprox} which usually are
asymptotically normal, and consistent estimators for $e_n/d_n$ times
the integral there.

We will outline how to estimate the term $I_2:=(e_n/d_n)
\int_{x_l}^\infty q(u)\eta(u,q(u)) \,\mathrm{d}u$, that is needed in the
case $\gamma_2 \ge0$. To avoid the estimation of the density $\eta$ of
$\nu$, we approximate the integral by the $\nu$-measure of a
shrinking set as follows. Because $\eta$ is continuous, for small
$\ell_n$ one has
\begin{eqnarray*}
\frac{e_n}{d_n}\int_{x_l}^\infty q(u)\eta
\bigl(u,q(u)\bigr)\, \mathrm{d}u &\approx& \frac{e_n}{d_n} \int_{x_l}^\infty
\frac{1}{2\ell_n}\int_{(1-\ell_n)q(u)}^{(1+\ell_n)q(u)} \eta(u,v)\, \mathrm{d}v \,\mathrm{d}u \\
&=&
\frac{1}{2\ell_n} \bigl(\nu\bigl(S_{n,2}^-\bigr)-\nu
\bigl(S_{n,2}^+\bigr)\bigr)
\end{eqnarray*}
with
\[
S_{n,2}^\pm:= \biggl\{ \frac{d_n}{e_n}\bigl(u,(1\pm
\ell_n)v\bigr)\Bigm|(u,v)\in S \biggr\}.
\]
Now one can proceed similarly as in \eqref{eq:pnapprox1} (using
\eqref{eq:nuapprox} and \eqref{eq:Sdef}) to construct an estimator
of $(e_n/d_n) \int_{x_l}^\infty q(u)\eta(u,q(u)) \,\mathrm{d}u$:

%
%co2.2 #&#
\begin{corollary} \label{cor:asvarest}
Let
\[
\hat S_{n,2}^\pm:= \biggl\{ \bigl(u,(1\pm
\ell_n)v\bigr)\Bigm|(u,v)\in\frac{n}{ke_n}\hat T_n^\leftarrow(D_n)
\biggr\}
\]
for some sequence $\ell_n\downarrow0$ such that
$k^{-1/2}(w_n(\gamma_1)\vee w_n(\gamma_2))=\mathrm{o}(\ell_n)$.
Suppose that all conditions of Theorem~\ref{theo:main} are
fulfilled and, in addition, that an analog to condition \textup{(D1)} holds
where $c_n$ is replaced with $c_n/(1\pm\ell_n)$.
Then
\[
\hat I_{n,2} := \frac{\hat\nu_n(\hat S_{n,2}^-)-\hat\nu_n(\hat
S_{n,2}^+)}{2\ell_n} = \frac{e_n}{d_n} \int
_{x_l}^\infty q(u)\eta\bigl(u,q(u)\bigr) \,\mathrm{d}u \bigl(1+
\mathrm{o}_P(1)\bigr).
\]
\end{corollary}

In a completely analogous way one can estimate $I_1:=(e_n/d_n)
\int_{q(\infty)}^\infty q^\leftarrow(v)\eta(q^\leftarrow(v),v)
\,\mathrm{d}v$ by $\hat I_{n,1} := (\hat\nu_n(\hat S_{n,1}^-)-\hat\nu_n(\hat
S_{n,1}^+))/(2\ell_n)$ with
\[
\hat S_{n,1}^\pm:= \biggl\{ \bigl((1\pm
\ell_n)u, v\bigr)\Bigm|(u,v)\in\frac{n}{ke_n}\hat
T_n^\leftarrow(D_n) \biggr\}.
\]

Now suppose that both extreme value indices $\gamma_i$ are positive
and that we estimate them by the Hill estimator, that is, $\hat\gamma_1
= k_1^{-1} \sum_{i=1}^{k_1} \log(X_{n-i+1:n}/X_{n-k_1:n})$ with
$X_{n-i+1:n}$ denoting the $i$th largest order statistic among
$X_1,\ldots,X_n$, and likewise $\hat\gamma_2 = k_2^{-1}
\sum_{i=1}^{k_2} \log(Y_{n-i+1:n}/Y_{n-k_2:n})$. It is well known
that
$k_i^{1/2}(\hat\gamma_i-\gamma_i)\to\mathcal{N}_{(0,\gamma_i^2)}$ if
condition (M2) holds and $k_i^{1/2} A_i(n/k_i)\to0$. In particular,
$\Gamma_i=0$ if $k=\mathrm{o}(k_i)$. However, if $k_i/k\to
\kappa_i\in(0,\infty)$ for both $i=1$ and $i=2$, then the joint
distribution of $\Gamma_1$ and $\Gamma_2$ is needed for the
construction of confidence intervals.

In the case $k_1=k_2=k$, de Haan and Resnick \cite{hr93} derived a
representation of $\Gamma_i$ in terms of a Gaussian process under
slightly different conditions than used in the present paper. One
may mimic their approach to show that under our conditions,
$(\Gamma_i/\gamma_i)_{i\in\{1,2\}}$ has the same distribution as
$ (  (\int_1^\infty t^{-1} W_i(t/\kappa_i)
\,\mathrm{d}t-W_i(1/\kappa_i) )/\kappa_i )_{i\in\{1,2\}}$ where
$(W_1,W_2)$ is a bivariate centered Gaussian process with covariance
function given by $\Cov(W_1(s),W_1(t))=\nu ((s\vee t,\infty)\times
(0,\infty) )$, $\Cov(W_2(s),W_2(t))=\nu ((0,\infty)\times
(s\vee
t,\infty)  )$ and
$\Cov(W_1(s),W_2(t))=\nu ((s,\infty)\times(t,\infty)  )$.
Direct calculations show that thus
$(\Gamma_i/\gamma_i)_{i\in\{1,2\}}$ is a centered Gaussian vector
with marginal variances $1/\kappa_i$ and covariance
$\nu ((\kappa_2,\infty)\times(\kappa_1,\infty) )$.
Hence, with $z_{1-\alpha/2}$
denoting the standard normal $(1-\alpha/2)$-quantile and
$\hat\sigma^2:= \hat I_{n,1}^2/\kappa_1+\hat I_{n,2}^2/\kappa_2+
2\hat\nu_n ((\kappa_2,\infty)\times(\kappa_1,\infty)
)\hat
I_{n,1}\hat I_{n,2}$,
%
%e2.13 #&#
\begin{equation}
\label{eq:confint} \bigl[ \hat p_n-k^{-1/2}e_n^{-1}
\log c_n \hat\sigma z_{1-\alpha/2}, \hat p_n+k^{-1/2}e_n^{-1}
\log c_n \hat\sigma z_{1-\alpha/2} \bigr]
\end{equation}
is a two-sided confidence interval for $p_n$ with asymptotic confidence
level $1-\alpha$. (This formula is also applicable if one of the
$\kappa_i$ equals $\infty$.)

As an alternative to the above approach, one may estimate the
density of the spectral measure $\Phi$ (cf. Cai \textit{et al}. \cite{ceh11}) and
construct both an estimator for the integrals and for the joint
distribution of the limiting random variables on the right-hand side
of \eqref{eq:mainapprox} from it.

%s2.5 #&#
\subsection{Choice of the blow-up factor} \label{subsect:ken}

Our estimation procedure consists of two steps. First the marginal
parameters are estimated using a certain fraction of largest order
statistics, and both the observations and the failure set are
marginally standardized accordingly. In the second step the transformed
failure set is blown up by a factor chosen by the statistician, and the
failure probability is estimated by a suitable fraction of the
empirical probability of the inflated set. As the choice of a suitable
sample fraction used in the marginal fitting has been extensively
discussed in literature (see, e.g., Beirlant \textit{et al}. \cite{bgstwf04},
Section~5.8), here we discuss how to choose the blow-up factor in the second
step. For simplicity, in the concrete calculations we focus on the case
that both extreme value indices are positive, but the general remarks
apply to the other cases as well.

The estimation of the marginal parameters $\gamma_i, a_i(n/k)$ and
$b_i(n/k)$ yield approximations of the marginal distribution functions
of the type
%
%e2.14 #&#
\begin{equation}
\label{eq:marginalfits} \hat F_i(x) := 1- \biggl( 1+\hat
\gamma_i\frac{x-\hat\mu_i}{\hat\sigma_i} \biggr)^{-1/\hat
\gamma_i},\qquad  i=1,2,
\end{equation}
which are sufficiently accurate for $x$ satisfying $1-F_i(x)\le k_i/n$.
The corresponding estimator $\hat U_i^\leftarrow:= 1/(1-\hat F_i)$ can
also be
interpreted as an estimator $(n/k)T_{n,i}^\leftarrow$ for different
values of $k$. However, if one starts with a given approximation of the marginal
tails as in \eqref{eq:marginalfits}, then the number $k$ does not have
any operational meaning. In that case it seems more
natural to reformulate our estimator $\hat p_n$, the main result
\eqref{eq:mainapprox} and the resulting confidence interval
\eqref{eq:confint} in terms of $\hat U_i^\leftarrow$.

For a fixed estimator $\hat U$ of $U$, the estimator of the failure probability
\[
\hat p_n = \frac{1}{e_n} \hat\nu_n \biggl(
\frac{n}{ke_n} \hat T_n^\leftarrow(D_n) \biggr) =
\frac{1}{ke_n} \sum_{i=1}^n
\eps_{\hat
U^\leftarrow(X_i,Y_i)} \biggl( \frac{n}{ke_n} \hat U^\leftarrow
(D_n) \biggr)
\]
depends on the constants $k$ and $e_n$ only via their product $ke_n$. At
first glance, this seems peculiar, because in Theorem~\ref{theo:main} the estimation error seemingly depends on $k$ and
$e_n$ in completely different ways. However, according to
the discussion in Section~\ref{subsect:confint}, for
$\gamma_1,\gamma_2>0$, approximation \eqref{eq:mainapprox} can be
rewritten as
%
%e2.15 #&#
\begin{equation}
\label{eq:errorapprox} \hat p_n-p_n = (ke_n)^{-1/2}
\log\frac{ke_n}n N \bigl(1+\mathrm{o}_P(1)\bigr)
\end{equation}
for a centered Gaussian random variable $N$ with variance
\begin{eqnarray*}
\sigma_N^2 & = & \frac{1}{e_n} \biggl(
\frac{k}{k_1} I_1^2 + \frac{k}{k_2}
I_2^2 + 2\nu \biggl( \biggl(\frac{k_2}k,\infty
\biggr)\times \biggl(\frac{k_1}k,\infty \biggr) \biggr)I_1
I_2 \biggr)
\\
& = & \frac{ke_n}{k_1} \biggl(\frac{I_1}{e_n} \biggr)^2 +
\frac{ke_n}{k_2} \biggl(\frac{I_2}{e_n} \biggr)^2 + 2\nu \biggl(
\biggl(\frac
{k_2}{ke_n},\infty \biggr)\times \biggl(\frac{k_1}{ke_n},\infty
\biggr) \biggr)\frac{I_1}{e_n} \frac{I_2}{e_n},
\end{eqnarray*}
where $I_1:= (e_n/d_n)\int_{q(\infty)}^\infty q^\leftarrow
(v)\eta (q^\leftarrow(v),v )\, \mathrm{d}v$ and $I_2:=
(e_n/d_n)\int_{x_l}^\infty q(u)\eta (u,q(u) ) \,\mathrm{d}u$. Thus
$I_i/e_n$ does not depend on $e_n$, and the distribution of the
approximating Gaussian random variable on the right-hand side of
\eqref{eq:errorapprox} depends on $k$ and $e_n$ only via their
product.

Moreover, also the estimators
\begin{eqnarray*}
\frac{\hat I_{n,1}}{e_n} & = & \frac{\hat\nu_n(\hat S_{n,1}^-)-\hat
\nu_n(\hat
S_{n,1}^+)}{2\ell_n e_n}
\\
& = & \frac{1}{2\ell_n}\cdot\frac{1}{ke_n} \sum
_{i=1}^n \eps_{\hat
U^\leftarrow(X_i,Y_i)} \biggl( \biggl\{
\bigl((1-\ell_n)u,v\bigr)\mid(u,v)\in \frac{n}{ke_n} \hat
U^\leftarrow(D_n) \biggr\} \\
&&\hphantom{\frac{1}{2\ell_n}\cdot\frac{1}{ke_n} \sum
_{i=1}^n \eps_{\hat
U^\leftarrow(X_i,Y_i)} \biggl(}{}\Bigm\backslash
 \biggl\{ \bigl((1+\ell_n)u,v\bigr)\mid(u,v)\in \frac{n}{ke_n}
\hat U^\leftarrow(D_n) \biggr\} \biggr)
\end{eqnarray*}
and likewise $\hat I_{n,2}/e_n$ depend on the product $ke_n$ only.
Finally, the covariance term $\nu (( k_2/(ke_n),\allowbreak \infty)\times
(k_1/(ke_n),\infty) )
= k^2 e_n/(\lambda k_1 k_2) \nu ((k/(\lambda k_1),\infty)\times
(k/(\lambda k_2),\infty) )$
can be estimated by
\begin{eqnarray*}
&& \frac{k^2 e_n}{\lambda k_1 k_2} \hat\nu_n \biggl( \biggl(\frac{k}{\lambda k_1},
\infty \biggr)\times \biggl(\frac{k}{\lambda
k_2},\infty \biggr) \biggr)
\\
&&\quad =
\frac{k e_n}{\lambda k_1 k_2} \sum_{i=1}^n
\eps_{\hat U^\leftarrow(X_i,Y_i)} \biggl( \biggl(\frac{n}{\lambda
k_1},\infty \biggr)\times
\biggl(\frac{n}{\lambda k_2},\infty \biggr) \biggr).
\end{eqnarray*}
Here the choice $\lambda\in(0,1]$ ensures that $\hat U^\leftarrow$
is used only on the range where it is a sufficiently accurate
estimator of the true function $U^\leftarrow$.

To sum up, all estimates only depend on $ke_n$, but not on the
numbers $k$ and $e_n$ separately. This product should be chosen as
large as possible under the constraints that both marginal
approximations of $U_i^\leftarrow$ by $\hat U_i^\leftarrow$ and the
approximation of the joint distribution of the standardized vector
(cf. \eqref{eq:nudef}) are reliable. To ensure the former
constraint, for the vast majority of the observations $(X_i,Y_i)$,
the indicator of the set $\{\hat U^\leftarrow(X_i,Y_i)\in n/(ke_n)
\hat U^\leftarrow(D_n)\}$ should not depend on the particular
values of $\hat U_1^\leftarrow(X_i)$ or $\hat U_2^\leftarrow(Y_i)$
if these are smaller than $n/k_1$ or $n/k_2$ (either because the
other component of the vector is so large that the observations lie
in the failure set anyway, or because the other component is so
small so that the indicator is 0 even if the maximal value $n/k_i$
is attained). For instance, if we consider failure sets of the type
$D_n:=\{(x,y)\mid\alpha_1 x+\alpha_2 y>R\}$, then $k
e_n$ should be smaller than $\min_{i=1,2} k_i\hat U_i^\leftarrow
(R/\alpha_i)$, because otherwise for sure $\hat U^\leftarrow(x,y)\in
(n/ke_n)\hat U^\leftarrow(D_n)$ for some values $(x,y)$ for which
$\hat U^\leftarrow(x,y)$ is not a reliable estimate of
$U^\leftarrow(x,y)$.

However, the above crude upper bound for $ke_n$ is not sufficient to
ensure that $\hat p_n$ is a reliable estimate of $p_n$, because the
dependence structure must be accurately described by the exponent
measure $\nu$, too. To determine a range of reasonable values for
$ke_n$, we propose (in analogy to the well-known Hill
plot used for selecting a reasonable sample fraction in the marginal
fitting), to plot $\hat p_n$ versus $ke_n$ and then to choose $ke_n$ in
a range where this curve seems stable. In the data example discussed in
Section~\ref{sect:data}, this approach seems to work pretty well.
(Motivated by the discussion by Drees \textit{et al}. \cite{dhr00}, it might also be
worthwhile to use a log-scale for $k e_n$ in order to get a clearer
picture about a good choice for this factor, but (unlike for the
so-called AltHill plot) a sound theoretical justification for this
modification is yet lacking.)

%s2.6 #&#
\subsection{Generalization to higher dimensions} \label{subsect:multi}
We conclude this section by indicating how to generalize the main
result to $\R^d$-valued vectors
$\mathbf{X}_i=(X_{i,1},\ldots,X_{i,d})$ of arbitrary dimension $d\ge
2$, albeit a detailed discussion is beyond the scope of this paper.
An inspection of the proof of Lemma~\ref{lemma:doubleint} reveals
that the generalized inverse $q^\leftarrow$ of the function $q$
is
used to describe the boundary of the set $S$ as a function of the
second coordinate. If $d>2$ (and hence the generalized inverse is
not defined), then an analogous description is needed for all
coordinates, that is, we need $d$ different representations of the set
$S$ of the form
%
%e2.16 #&#
\begin{equation}
\label{eq:Srep} S = \bigl\{ \mathbf{x}\in[0,\infty)^d\mid
x_i\ge q_i(\mathbf {x}_{-i}) \bigr\},\qquad  1\le
i\le d,
\end{equation}
where $\mathbf{x}_{-i}\in[0,\infty)^{d-1}$ denotes the vector
$\mathbf{x}$ with $i$th coordinate removed and $q_i$ are suitable
$[0,\infty]$-valued functions that are decreasing in each argument.
Then one may proceed as in the case $d=2$ by separately examining
the influence of the transformation of each marginal on the
$\nu$-measure of the (suitably restricted) set $S$. Under suitable
integrability conditions on the functions $q_i$ and obvious
generalizations of the conditions (M1)--(M3), (D1), (D2) and
(S1)--(S3), it can be shown that
%
%e2.17 #&#
\begin{eqnarray}\label{eq:mainmulti}
\hspace*{-15pt}&&k^{1/2}d_n(\hat p_n-p_n)
\nonumber
\\
&&\quad  =  \sum_{i=1}^d w_n(
\gamma_i)\lleft\{ %
\begin{array} {l@{ \qquad }l} -
\displaystyle \frac{\Gamma_i}{\gamma_i} \displaystyle \int q_i(v)\eta\bigl(\tilde q_i(v)
\bigr) 1_{(0,\infty)}\bigl(q_i(v)\bigr) \llambda^{d-1}(\mathrm{d}v),
& \gamma_i>0,
\\
\biggl(\displaystyle \frac{\alpha_i}{\gamma_i}-\beta_-\displaystyle \frac{\Gamma_i}{\gamma
_i^2} \biggr) \\
\quad {}\times\displaystyle \int
\bigl(q_i(v)\bigr)^{1-\gamma_1} \eta\bigl(\tilde
q_i(v)\bigr) 1_{(0,\infty)}\bigl(q_i(v)\bigr)
\llambda^{d-1}(\mathrm{d}v), & \gamma_i<0,
\\
- \Gamma_i \displaystyle \int q_i(v)\eta\bigl(\tilde
q_i(v)\bigr) 1_{(0,\infty)}\bigl(q_i(v)\bigr)
\llambda^{d-1}(\mathrm{d}v), & \gamma_i=0 \end{array} %
\rright.
\\
& &\qquad {} +\mathrm{o}_P\bigl(w_n(\gamma_i)
\bigr). \nonumber
\end{eqnarray}
Here $\llambda^{d-1}$ denotes the Lebesgue measure on
$[0,\infty)^{d-1}$ and $\tilde q_i(v)$ is the vector in
$[0,\infty)^d$ whose $i$th coordinate equals $q_i(v)$ and the other
$d-1$ coordinates are those of $v$.

If the boundary $\partial S$ of the set $S$ is sufficiently smooth,
then the integrals on the right-hand side of \eqref{eq:mainmulti}
can be represented more naturally as integrals w.r.t. certain
differential forms (see, e.g., Schreiber \cite{s77}, for an informal
introduction to differential forms). More precisely, assume that
there exists a set $D\subset[0,\infty)^{d-1}$ and a continuously
differentiable function $q\dvtx D\to[0,\infty)$, such that $\partial
S=\{ \Psi(u):=(u,q(u))\mid u\in D\}$.
Then the right-hand side of \eqref{eq:mainmulti} equals
\begin{eqnarray*}
&& \sum_{i=1}^d w_n(
\gamma_i)\lleft\{ %
\begin{array} {l@{\qquad  }l} -
\displaystyle \frac{\Gamma_i}{\gamma_i} \int_\Psi \mathit{pr}_i\cdot\eta \,\mathrm{d}
x_1\wedge\cdots\wedge \,\mathrm{d}x_{i-1}\wedge \,\mathrm{d}x_{i+1}
\wedge\cdots\wedge \,\mathrm{d}x_d, & \gamma_i>0,
\nonumber
\\
\biggl(\displaystyle \frac{\alpha_i}{\gamma_i}-\beta_-\displaystyle \frac{\Gamma_i}{\gamma
_i^2} \biggr)\\
\quad {}\times \displaystyle \int
_\Psi(\mathit{pr}_i)^{1-\gamma_i}\cdot\eta \,\mathrm{d}
x_1\wedge\cdots\wedge \,\mathrm{d}x_{i-1}\wedge \,\mathrm{d}x_{i+1}
\wedge\cdots\wedge \,\mathrm{d}x_d, & \gamma_i<0,
\\
- \Gamma_i \displaystyle \int_\Psi \mathit{pr}_i\cdot
\eta \,\mathrm{d} x_1\wedge\cdots\wedge \,\mathrm{d}x_{i-1}\wedge
\,\mathrm{d}x_{i+1}\wedge\cdots\wedge\, \mathrm{d}x_d, & \gamma_i=0
\end{array} %
\rright.
\\
&&\quad {} +\mathrm{o}_P\bigl(w_n(\gamma_i)\bigr)
\end{eqnarray*}
with $\textit{pr}_i$ denoting the projection to the $i$th coordinate, which is
the integral of a $(d-1)$-form over $\delta S$. This
representation reflects most clearly the fact that the $i$th term
results from the change of the boundary surface of $S$ by the
marginal transformation $H_{n,i}$. Such
a representation can be derived for more general differentiable
manifolds $\partial S$.

%s3 #&#
\section{Analysis of insurance claims} \label{sect:data}

In this section, we discuss issues arising in the analysis of a
well-known data set of claims to Danish
fire insurances. The data set contains losses to building(s), losses
to contents and losses to profits (caused by the same fire) observed
in the period 01/1980--12/2002, discounted to 07/1985. The claims
are recorded only if the sum of all components exceeds 1 million
Danish Kroner (DKK). Due to this recording method, there
is an artificial negative dependence between the components, since
if one component is smaller than 1 million DKK, the sum of the
others must be accordingly larger. To avoid this effect, we
therefore consider only those claims for which at least one
component exceeds 1 million DKK, which leads to a sample of 3976
claims. Moreover, we focus on the losses to buildings, denoted by
$X_i$ as a multiple of one million DKK, and the losses to contents
$Y_i$. A more detailed description of the data can be found in
M\"{u}ller \cite{m08} and Drees and M\"{u}ller \cite{dm08}.

%
%f1 #&#
\begin{figure}[b]

\includegraphics{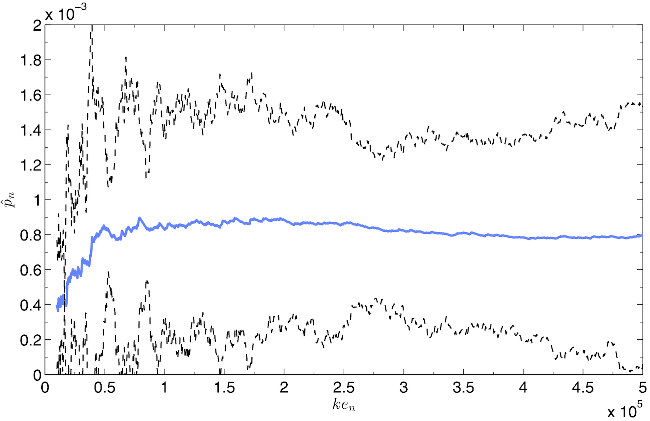}

\caption{$\hat p_n$ (solid  line) and confidence intervals (dashed line) versus $ke_n$ for
Danish fire insurance claims.} \label{fig:probabplot}
\end{figure}

As described in the \hyperref[intro]{Introduction}, we assume that a quota reinsurance
pays $(1-\alpha_X) X_i$ for each loss $X_i$ to the building and
$(1-\alpha_Y)Y_i$ for each loss $Y_i$ of content, while an
XL-reinsurance pays if the remaining costs $\alpha_X X_i+\alpha_Y
Y_i$ exceed a retention level $R$. We want to estimate the
probability $p_n:=P(D_n)$ with $D_n:=\{\alpha_X X_i+\alpha_Y
Y_i>R\}$ that a fire results in a claim to the XL-reinsurance for
$\alpha_X=1$, $\alpha_Y=0.5$ and $R=100$. (More
precisely, we estimate the conditional probability given that
$\max(X_i,Y_i)>1$.)

M\"{u}ller \cite{m08}, Section~5.1.2, fitted GPD's to the
marginal distributions using the Hill estimators based on the
$k_1=900$ and $k_2=600$ largest observations to obtain a tail
approximation of the type \eqref{eq:marginalfits}
with parameters $\hat\gamma_1=0.57$, $\hat\sigma_1=0.54$,
$\hat\mu_1=0.91$, $\hat\gamma_2=0.72$, $\hat\sigma_2=0.47$ and
$\hat\mu_2=0.15$. Moreover, he showed that the
components of the claim vector are apparently asymptotically
dependent.

As suggested in Section~\ref{subsect:ken}, in Figure~\ref{fig:probabplot}
we plot the estimate of the failure probability versus $ke_n$ for
values of $ke_n$
ranging from $10^4$ to $5\cdot10^5$. In the present situation, the
crude upper bound on $ke_n$ discussed in Section~\ref{subsect:ken}
is about $1.7\cdot10^6$, but the curve
of probability estimates shows a clear downward trend for
$ke_n>2\cdot10^5$, which is most likely due to a deviation of the
dependence structure from its limit. On the other hand, for values
smaller than $5\cdot10^4$ the curve is very unstable, too, because
the random error is too large as just a few observations fall
into the inflated failure set (e.g., about 25 if $ke_n\approx3\cdot
10^4$). This lower bound on $ke_n$ reflects the condition in the
asymptotic framework that $n$ is of smaller order than $k e_n$
(see condition (S1)). In view of this plot, the choice $ke_n=2\cdot
10^5$ seems reasonable.

In addition, Figure~\ref{fig:probabplot}
shows a two-sided
confidence interval with nominal size $0.95$, again as a function of
$ke_n$. Here we have chosen $\ell_n=0.1$ and $\lambda=1$ in the
estimator of the variance $\sigma^2_N$ described above; other values
of $\lambda$ between $1/2$ and 1 yield approximately the same
estimates, while smaller values of $\ell_n$ lead to larger
fluctuations in the confidence bounds, that however are still of a
similar size.

For $ke_n=2\cdot10^5$ one obtains a point estimate for $p_n$ of
about $8.8\cdot10^{-4}$ and a confidence interval $[2.2\cdot
10^{-4}, 1.54\cdot10^{-3}]$. At first glance, this confidence
interval seems rather wide. However, here we estimate the probability
of a very rare event which has
occurred only twice in the observational period of more than 20
years. Indeed the empirical probability of the event is about
$5\cdot10^{-4}$, and the Clopper--Pearson confidence interval
$[6\cdot10^{-5}, 1.8\cdot10^{-3}]$ (again with nominal size
$0.95$) is even wider. It is worth mentioning that both the
empirical point estimate and the Clopper--Pearson confidence interval
are exactly the same if one wants to estimate the probability that a
claim occurs to the XL-reinsurance for any retention level $R$
between 77 and 145 million DKK! Moreover, for retention level above
152 million DKK the point estimate would be 0 and thus useless for
purposes of risk management.

%s4 #&#
\section{Simulation study} \label{sect:simus}

In a small simulation study, we examine the finite sample behavior of
our estimator of a failure probability. In particular, we want to
compare its performance with that of the estimator proposed by
de Haan and Sinha \cite{hs99}. Moreover, we demonstrate that often
the fit of the
marginal distributions is the main source of the random error, as
indicated by the main Theorem~\ref{theo:main}.

We consider two different models for the dependence structure of $(X,Y)$.
\begin{itemize}
\item For the first model class we assume that $(X,Y)$ has a Gumbel
copula, that is, for some $\th\in(1,\infty)$
\[
F\bigl(F_1^\leftarrow(u),F_2^\leftarrow(v)
\bigr) = C_\th^{\mathrm{Gum}}(u,v) = \exp \bigl(- \bigl(|\log
u|^\th+|\log v|^\th \bigr)^{1/\th} \bigr),\qquad  0< u,v
\le1.
\]
Note that $C_\th^{\mathrm{Gum}}$ is the Copula of a bivariate extreme value
distribution. The corresponding exponent measure is given by
\[
\nu_\th^{\mathrm{Gum}} \bigl((u,\infty)\times(v,\infty) \bigr) =
\frac{1}u+\frac{1}v + \bigl(u^{-\th}+v^{-\th}
\bigr)^{1/\th}, \qquad u,v>0,
\]
and the spectral density by
\[
\varphi_\th^{\mathrm{Gum}}(\arctan t) = (\th-1) \bigl(1+t^2
\bigr)^{3/2} t^{-(\th+1)} \bigl(1+t^{-\th}
\bigr)^{1/\th-2}.
\]
Hence,
\[
\varphi_\th^{\mathrm{Gum}}(t) \sim(\th-1)t^{\th-2} \quad \mbox{and} \quad \varphi _\th ^{\mathrm{Gum}} (\uppi/2-t ) \sim(\th-1)t^{\th-2}\qquad
\mbox{as } t\downarrow0,
\]
and condition (D2) is obviously fulfilled for $\th\ge2$, whereas for
$\th\in(1,2)$ the spectral density is not bounded and thus (D2) is
not satisfied.

In the simulation results presented below, the Gumbel copula with $\th
=5$ is combined with generalized extreme value (GEV) marginals
$F_i(x)=\exp (-(1+\gamma_i x)^{-1/\gamma_i} )$ for $1+\gamma
_i x>0$ with extreme value index $\gamma_1=\gamma_2\in\{-1/4,0,1/4\}
$. (Simulations with GPD marginals yielded very similar results which
are thus not presented here.)

\item The second model class is related to the model suggested by
Schlather \cite{s02}. Similarly as in Example~3.8 of Segers \cite
{s12}, we define
$(X,Y)=\sqrt{2\uppi} (SZ,TZ)$ where $Z$ is a unit Fr\'{e}chet random
variable (i.e., $P\{Z\le x\}=\exp(-1/x)$ for $x>0$) independent from
the vector $(S,T)$ which has the distribution of a centered normal
vector with variances 1 and covariance $\varrho\in(-1,1)$ conditioned
on both coordinates being positive. By Lemma~3.1 of Segers
\cite{s12}, the
distribution of $(X,Y)$ belongs to the max domain of attraction of an
extreme value distribution with unit Fr\'{e}chet marginals. Direct
calculations show that the stable tail dependence function is given by
\begin{eqnarray*}
\ell(x,y)&:=&\nu \biggl( \biggl(\frac{1}x,\infty \biggr) \times \biggl(
\frac{1}y,\infty \biggr) \biggr)\\
& = &\frac{1}{1+\rho} \bigl(\rho (x+y)+
\bigl(x^2+y^2-2\rho xy\bigr)^{1/2} \bigr),\qquad  x,y
\ge0.
\end{eqnarray*}
Hence, the exponent measure has the density
\[
\eta(u,v)=\frac{1-\varrho^2}{1+\rho} \bigl(u^2+v^2-2\varrho uv
\bigr)^{-3/2},
\]
and the pertaining spectral density
\[
\varphi(t) = \frac{1-\varrho^2}{1+\rho} \biggl(1-2\varrho\frac
{\tan t}{1+\tan^2 t} \biggr)
\]
is strictly positive and continuous on $[0,\uppi/2]$, so that condition
(D2) is fulfilled. We have simulated data sets with correlation
$\varrho\in\{-0.8, 0.2, 0.8\}$, but for briefness sake will report
results only for the last value, because they look similar in all three cases.

Direct calculations show that the marginal distributions are symmetric with
\[
1-F_i(x)=1/2-\exp \bigl(\uppi x^{-2} \bigr) \bigl(1-\Phi
\bigl(\sqrt{2\uppi }x^{-1} \bigr) \bigr),\qquad  x>0.
\]
\end{itemize}

We want to estimate failure probabilities of linear half spaces of the
form $D=\{(x,y)\mid x+y/2>R\}$ where $R$ has been chosen such that the
failure probability $p_n$ (which is determined by simulations) lies
between $2\cdot10^{-4}$ and $5\cdot10^{-4}$ for each of the above models.

For each setting, we have simulated 1000 data sets of size $n=500$. As
$n p_n\ll1$ (with a value of about 0.1 in most settings), the
estimation of $p_n$ is a challenging task. In particular, one cannot
use empirical probabilities, because in most simulations the failure
set will not contain any data point.

To make a comparison with the estimator $\hat p_n^{(\mathrm{HS})}$ proposed by
de Haan and Sinha \cite{hs99} easier, we use their marginal estimators:
\begin{eqnarray*}
\hat\gamma_1 & := & M_1(X)+1-\frac{1}2 \bigl/
\biggl(1-\frac
{M_1^2(X)}{M_2(X)} \biggr),
\\
\hat a_1(n/k) & := & X_{n-k:n} \bigl(3M_1^2(X)-M_2(X)
\bigr)^{1/2}
\\
& &{} \times \biggl(\frac{3}{(1-\hat\gamma_1\vee0)^2}- \frac{2}{(1-\hat
\gamma_1\vee0)(1-2\hat\gamma_1\vee0)} \biggr)^{-1/2},
\\
\hat b_1(n/k) & := & X_{n-k:n}
\end{eqnarray*}
with
\[
M_r(X)=\frac{1}k \sum_{i=1}^k
\biggl(\log\frac
{X_{n-i+1:n}}{X_{n-k:n}} \biggr)^r.
\]
For small values of $k$ it happens (in at most a few percents of the
simulations) that the scale estimate $\hat a_1(n/k)$ is not defined. In
these simulations, instead we use the moment estimator
\[
\hat a_1(n/k) := \frac{1}2 X_{n-k:n}
M_1(X) \bigl/ \biggl(1-\frac
{M_1^2(X)}{M_2(X)} \biggr)
\]
proposed by Dekkers \textit{et al}. \cite{deh89}. The estimators of the parameters of
the second marginal distribution are defined likewise with $X_i$
replaced with $Y_i$.

We now discuss our findings for the model with Gumbel copula
$C_5^{\mathrm{Gum}}$ and Gumbel margins (i.e., GEV margins with $\gamma=0$) in
detail. The results for the other distributions of $(X,Y)$ are then
presented more briefly. In this model, the true failure probability
$p_n$ is about $2.25\cdot10^{-4}$ for $R=12$.

Figure~\ref{fig:RMSEGEV0} displays the empirical root mean squared
error (RMSE) of our estimator $\hat p_n$ as a function of the number
$k$ of largest order statistics used for marginal fitting and the
product $ke_n$ which determines the blow-up factor. As expected, if
either of these values is small, then the RMSE is high due to a large
standard error, while for too large a value the bias leads to a large
RMSE; in particular the first effect is much more pronounced for the
blow-up factor. It can be most clearly seen from the contour lines
shown in the left plot of Figure~\ref{fig:contGEV0} that values of $k$
in the range 50--150 and values of $ke_n$ around $1.5\cdot10^6$ yield
most accurate estimates.
%
%f2 #&#
\begin{figure}

\includegraphics{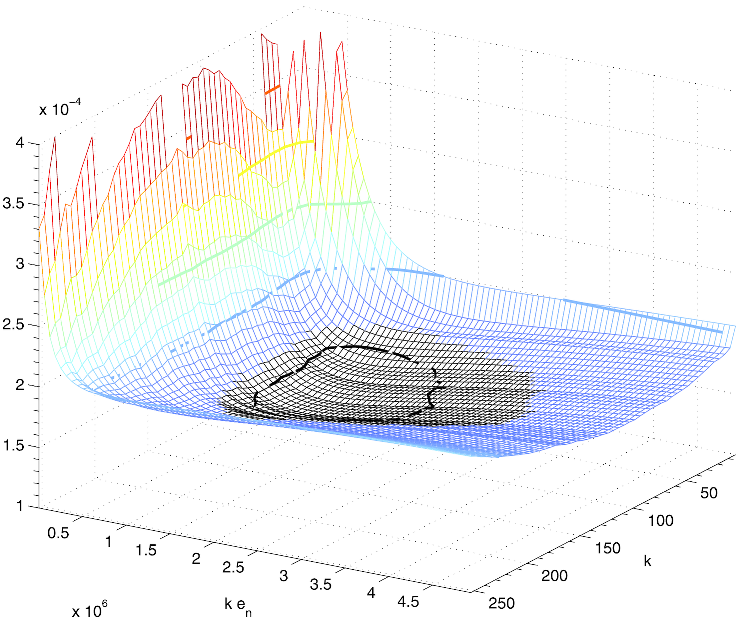}

\caption{Empirical RMSE of $\hat p_n$ for $(X,Y)$ with Gumbel $\th=5$
copula and Gumbel marginals.} \label{fig:RMSEGEV0}
\end{figure}

Next we compare the performance of $\hat p_n$ with that of the
estimator $\hat p_n^{(\mathrm{HS})}$ suggested by de Haan and Sinha \cite{hs99}.
Recall from Section~\ref{subsect:alternatives} that the main
difference between these estimators is that the latter uses a
data-driven value $\hat c_n$ instead of $ke_n/n$, which is constructed
as an estimator of the (quite arbitrarily fixed) factor $d_n$; see
formula (1.6) of de Haan and Sinha \cite{hs99}. The graph on the right-hand
side of Figure~\ref{fig:contGEV0} shows a contour plot of the ratio of
the RMSE of $\hat p_n^{(\mathrm{HS})}$ and of the RMSE of $\hat p_n$ again as a
function of $k$ and $ke_n$. Obviously, the RMSE of $\hat p_n^{(\mathrm{HS})}$ is
much larger than that of our estimator for almost all values of $k$ and
$e_n$. Indeed, for the most reasonable choices of $k$ the former is
usually as least double as large as the latter.
%
%f3 #&#
\begin{figure}

\includegraphics{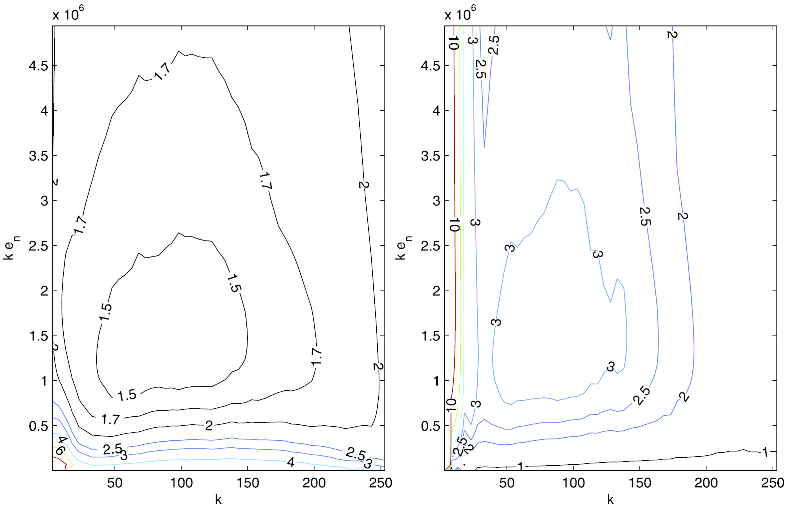}

\caption{Contour lines of $10^4\times$RMSE of $\hat p_n$ (left) and
of the ratio between the RMSE of $\hat p_n^{(\mathrm{HS})}$ and $\hat p_n$
(right).}\vspace*{15pt} \label{fig:contGEV0}
%\end{figure}
%
%
%f4 #&#
%\begin{figure}

\includegraphics{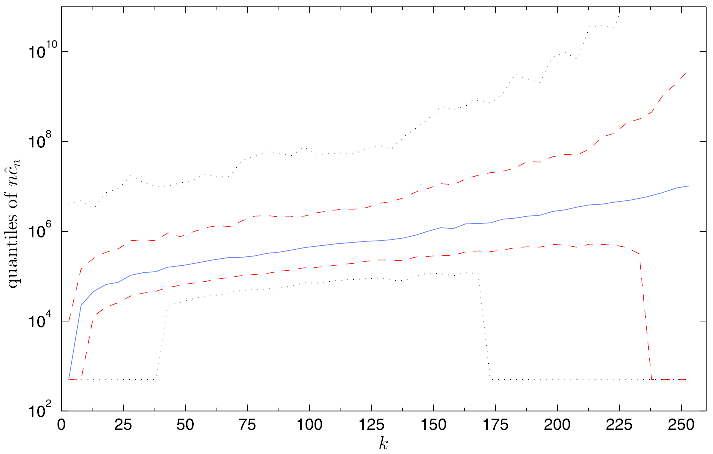}

\caption{Empirical quantiles of $n\hat c_n$ to the levels 0.1, 0.25,
0.5, 0.75 and 0.9.} \label{fig:cnhquantsGEV0}
\end{figure}

The reason for the inferiority of $\hat p_n^{(\mathrm{HS})}$ can be seen from
Figure~\ref{fig:cnhquantsGEV0} which shows empirical quantiles of
$n\hat c_n$ (corresponding to $ke_n$) as a function of $k$ for the
levels 0.1, 0.25, 0.5, 0.75 and 0.9. For $k<200$, in more than half of
the simulations $n\hat c_n$ is clearly smaller than $10^6$, whereas a
good choice of $k e_n$ would be in the range $10^6-2\cdot10^6$.
Moreover, the variation of $n\hat c_n$ is huge with more than 10\% of
the estimates exceeding $10^7$ for $k\ge50$. This figure indicates
that the estimator employed is rather inaccurate. However, even if the
theoretical value of $c_n=kd_n/n$ was known, in general the resulting
estimator of the failure probability would not perform well, because
$d_n$ is arbitrarily defined by the requirement that ($1,1$) lies on the
boundary of $S$. In particular, this choice of $d_n$ is not at all
related to the accuracy of the approximation \eqref{eq:nuapprox} which
strongly influences the part of the bias not resulting from the fitting
of the marginal distributions. For example, in the present situation
$d_n=U^\leftarrow(R/1.5)=\e^{R/1.5}\approx 3000$, which leads to much
too small values of $ke_n$ if one chooses $e_n=d_n$.

In the situation of Theorem~\ref{theo:main}, asymptotically the main
source of \emph{random} error is the fitting of the marginal
distributions. To check whether this bears out for moderate sample
sizes, we have calculated an analog to our estimator of the failure
probability where the marginal estimator $T_{n,i}^\leftarrow$ is
replaced with the true function $(k/n)U^\leftarrow$:
\[
\hat p_n^{(tm)} := \frac{1}{ke_n} \sum
_{i=1}^n \eps_{U^\leftarrow
(X_i,Y_i)} \biggl(
\frac{n}{ke_n} U^\leftarrow(D) \biggr).
\]
The left-hand plot in Figure~\ref{fig:GEV0cont00} displays contour
lines of the ratio between the standard errors of $\hat p_n$ and of
$\hat p_n^{(tm)}$ as a function of $k$ and $ke_n$. For almost all
values of these tuning parameters, the standard error of the estimator
$\hat p_n$ with estimated marginals is at least 5 times larger than
that of $\hat p_n^{(tm)}$, and it is more than 100 times larger if
$ke_n>1.5\cdot10^6$. However, if one considers the RMSE, then the
picture is less dramatic. While the total error of $\hat p_n$ is still
at least about 25\% higher than that of $\hat p_n^{(tm)}$, for the most
reasonable choices of $k$ and $e_n$ it is less than three times as
large (middle plot in Figure~\ref{fig:GEV0cont00}). The reason for
this different behavior of the RMSE is of course the bias, which is at
least 5 times higher in absolute value than the standard error for
$\hat p_n^{(tm)}$ with $ke_n>0.5\cdot10^6$ and is thus the dominating
part of the total error. At first glance one might expect that using
the true marginal distributions should also reduce the bias. However,
in general this is not true, since the bias resulting from the marginal
fit may partly cancel out with the bias caused by the error in
approximation \eqref{eq:nuapprox}. Indeed, for moderate values of
$ke_n$ the absolute bias of $\hat p_n$ is often smaller than that of
$\hat p_n^{(tm)}$ as it can be seen from the contour plot of the ratio
of these values in the right-hand plot of Figure~\ref{fig:GEV0cont00}.
%
%f5 #&#
\begin{figure}

\includegraphics{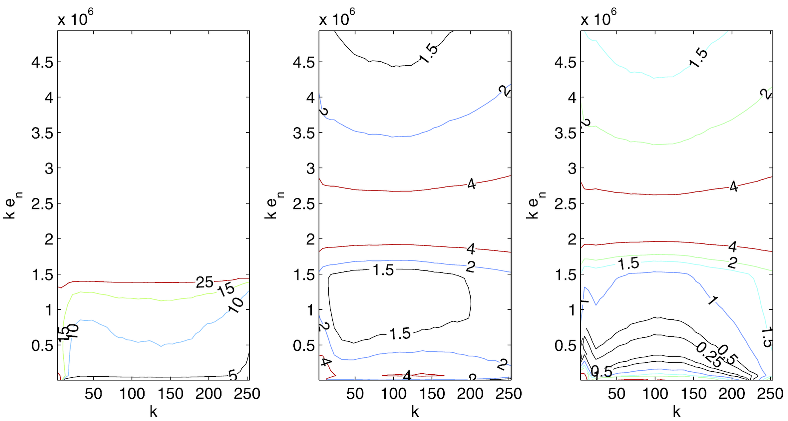}

\caption{Contour plots of the ratio of the standard errors of $\hat
p_n$ and of $\hat p_n^{(tm)}$ (left), of the corresponding ratios of
RMSE (middle) and of the ratio of absolute bias (right).} \label{fig:GEV0cont00}
\end{figure}

Table~\ref{tab:distrib} lists the true failure probabilities for the
four remaining settings not discussed so far.
%
%t1 #&#
\begin{table}[b]
\tablewidth=250pt
\tabcolsep=0pt
\caption{Failure probabilities for different models and thresholds
$R$} \label{tab:distrib}
\begin{tabular*}{250pt}{@{\extracolsep{\fill}}llll@{}}
\hline
\multicolumn{1}{l}{\multirow{2}{60pt}[-12pt]{Model}} & \multicolumn{2}{l}{Copula $C^{\mathrm{Gum}}_5$} & \multicolumn{1}{l}{\multirow{2}{60pt}[-6pt]{Schlather type
model $\varrho=0.8$}}\\[-5pt]
 & \multicolumn{2}{l}{\hrulefill} & \\
& \multicolumn{1}{l}{$\gamma=0.25$} & \multicolumn{1}{l}{$\gamma=-0.25$} & \\ \hline
$R$ & $40$ & $5$ & 5000 \\
$10^4\times$failure probab. & $\hphantom{4}{2.21}$ & $5.00$ & 3.40\\
\hline
\end{tabular*}
\vspace*{-4pt}
\end{table}

Figure~\ref{fig:plots} shows contour lines of the RMSE of $\hat p_n$
and of the ratio of the RMSE of $\hat p_n^{(\mathrm{HS})}$ and $\hat p_n$ in the
two left-hand columns (corresponding to the plots in Figure~\ref
{fig:contGEV0}) and the ratios of the standard error and the RMSE of
$\hat p_n$ and $\hat p_n^{(tm)}$ (corresponding to the first two plots
in Figure~\ref{fig:GEV0cont00}) in the two right-hand columns. The
four models are displayed in the four rows in the same order as in
Table~\ref{tab:distrib}.

While in most cases the results are qualitatively the same as for the
Gumbel model discussed above, there are two remarkable exceptions.
First, in the model with Gumbel copula and GEV marginals with $\gamma
=-0.25$ the estimator suggested by de Haan and Sinha \cite{hs99} works
pretty well with an RMSE which is at most 25\% larger than that of
$\hat p_n$ for most combinations of $k$ and $e_n$. (For some not very
reasonable combinations it is even smaller.)

Second, in the Schlather type model with $\varrho=0.8$, due to its
rather large bias, the RMSE of the estimator $\hat p_n^{(tm)}$ which
uses the true marginal distributions is clearly larger than the one of
$\hat p_n$ if $k\ge50$ and $ke_n$ is about $8\cdot10^5$. However, as
the first plot for this distribution shows, these combinations of
values for $k$ and $ke_n$ are not good choices for $\hat p_n$, because
its RMSE is more than 4 times as large as the minimal value. Indeed,
for all simulated distributions the bias of $\hat p_n^{(tm)}$ was more
stable than that of $\hat p_n$ (in particular for small value of $k$),
thus indicating again that one should be particularly careful with the
estimation of the marginal distributions.
%
%f6 #&#
\begin{figure}

\includegraphics[scale=0.95]{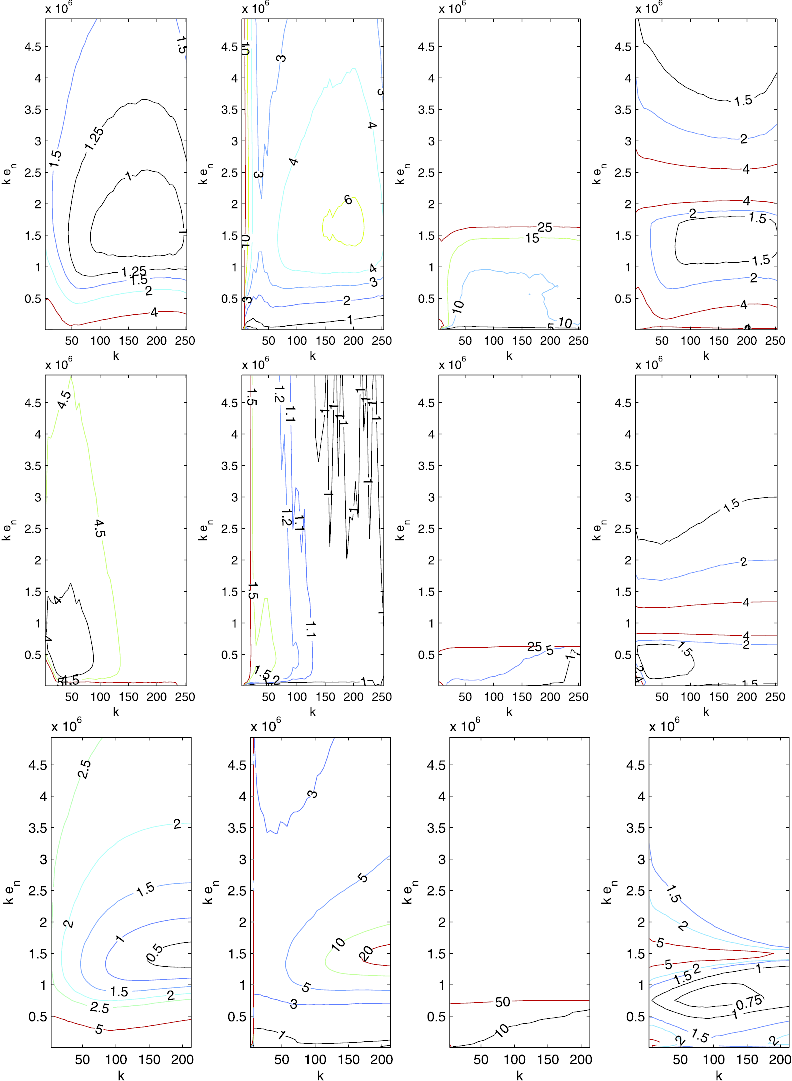}

\caption{From left to right: contour plots of $10^4\times$RMSE, the
ratio of RMSE of $\hat p_n^{(\mathrm{HS})}$ and $\hat p_n$, the ratio of
standard errors of $\hat p_n$ and $\hat p_n^{(tm)}$, and the ratio of
RMSE of $\hat p_n$ and $\hat p_n^{(tm)}$; from top to bottom: Gumbel
$\th=5$ copula with GEV marginals with $\gamma=0.25$ resp. $\gamma
=-0.25$, Schlather type model with $\varrho=0.8$.}
\label{fig:plots}
\end{figure}

%s5 #&#
\section{Proofs} \label{sect:proofs}

The proof of Theorem~\ref{theo:main} is based on the decomposition
\eqref{eq:esterror} of the estimation error. The asymptotic behavior
of the leading term $\mathit{IV}+V$ is established in Corollary~\ref{cor:termIVVapprox}. As a preparation for this result, first we
establish an approximation of the random transformation of the
marginals defined in \eqref{eq:Hndef}. Thereby we must restrict
ourselves to arguments that are neither too small nor too large,
which defines a certain subset $S_n^*$ of $S$. In Lemma~\ref{lemma:Snstarapprox} an upper bound on the difference between
the $\nu$-measures of $S$ and $S_n^*$ is derived, while the Lemmas
\ref{lemma:doubleint} and \ref{lemma:singleintapprox} analyze the
influence of the marginal transformations on the $\nu$-measure of
$S_n^*$.

In Lemma~\ref{lemma:empprocess} an upper bound on the term $\mathit{II}$ of
decomposition \eqref{eq:esterror} is proved using empirical process
theory. Finally, Lemma~\ref{lemma:truncInegl} establishes upper
bounds on the terms $I$ and $\mathit{III}$, while Lemma~\ref{lemma:termVI}
takes care of $\mathit{VI}$.

%
%le5.1 #&#
\begin{lemma} \label{lemma:marginalapprox}
Assume that the conditions \textup{(M1)--(M3)} and \textup{(S1)} are fulfilled.
For $i\in\{1,2\}$, let $ \lambda_{n,i}> 0$ be a decreasing and $\tau
_{n,i}<\infty$ an increasing sequence,
such that the following conditions are met:
\begin{enumerate}[(iii)]
\item[(i)]$ A_i(n/k)
(\lambda_{n,i}d_nk/n)^{\rho_i\pm\eps}=\mathrm{o} (k^{-1/2}
w_n(\gamma_i) )$ for some $\eps>0$.
\item[(ii)] If $\gamma_i>0$, then
$k^{-1/2}=\mathrm{o} ((\lambda_{n,i}d_n/e_n)^{\gamma_i} )$.
\item[(iii)] If $\gamma_i<0$, then $k^{-1/2}=\mathrm{o} ((\tau_{n,i}d_n
k/n)^{\gamma_i} )$ and
$\log(d_n/e_n)=\mathrm{o} ((d_nk/n)^{-\gamma_i} )$.
\item[(iv)] If $\gamma_i=0$, then $k^{-1/2}\log\tau_{n,i}\to
0$ and $\log(d_n/e_n)=\mathrm{o}(\log c_n)$.
\end{enumerate}
Then, for $i\in\{1,2\}$,
\begin{eqnarray*}
\frac{d_n}{e_n} H_{n,i}(x) &=& T_{n,i}^\leftarrow\circ
\hat T_{n,i} \circ \hat T_{n,i}^{(c)\leftarrow}\circ
U_i(d_n x)
\\
& = &\frac{d_n}{e_n} x \left(\vphantom{\begin{array} {l@{\qquad  }l} -k^{-1/2} \log
c_n \biggl(\displaystyle \frac{\Gamma_i}{\gamma_i} +\mathrm{o}_P(1)\biggr)+
\mathrm{O}_P\bigl(k^{-1/2}(xd_n/e_n)^{-\gamma_i}
\bigr), & \gamma_i>0,
\\
k^{-1/2} (d_nk/n)^{-\gamma_i} \bigl(\bigl(
\alpha_i/\gamma_i-\beta_i-
\Gamma_i/\gamma _i^2+\mathrm{o}_P(1)
\bigr)x^{-\gamma_i}+\mathrm{o}_P(1) \bigr), &
\gamma_i<0,
\\
-k^{-1/2} \log^2 c_n \bigl(\Gamma_i/2+
\mathrm{o}_P(1)\bigr)+ \mathrm{O}_P
\bigl(k^{-1/2} \log c_n\log x\bigr), & \gamma_i=0,
\end{array}} 1\right.
\\
& &\hphantom{\frac{d_n}{e_n} x \left(\right.}{} +\left.\lleft\{ %
\begin{array} {l@{\qquad  }l} -k^{-1/2} \log
c_n \biggl(\displaystyle \frac{\Gamma_i}{\gamma_i} +\mathrm{o}_P(1)\biggr)+
\mathrm{O}_P\bigl(k^{-1/2}(xd_n/e_n)^{-\gamma_i}
\bigr), & \gamma_i>0,
\\
k^{-1/2} (d_nk/n)^{-\gamma_i} \\
\quad {}\times\bigl(\bigl(
\alpha_i/\gamma_i-\beta_i-
\Gamma_i/\gamma _i^2+\mathrm{o}_P(1)
\bigr)x^{-\gamma_i}+\mathrm{o}_P(1) \bigr), &
\gamma_i<0,
\\   \noalign{\vspace*{2pt}}
-k^{-1/2} \log^2 c_n \bigl(\Gamma_i/2+
\mathrm{o}_P(1)\bigr)+ \mathrm{O}_P
\bigl(k^{-1/2} \log c_n\log x\bigr), & \gamma_i=0
\end{array} %
\rright. \hspace*{-3pt}\right)
\end{eqnarray*}
uniformly for $x\in[\lambda_{n,i},\tau_{n,i}]$.
\end{lemma}

\begin{pf}
For notational simplicity, we omit all indices and arguments of the marginal
parameters and normalizing functions and their estimators; for example,
we use $\hat a$ as a short form
of $\hat a_i(n/k)$. Moreover, we drop all indices referring to the
$i$th marginal, that is,
we write $U$ instead of $U_i$, $T_n$ instead of $T_{n,i}$ and so on.

By \eqref{eq:secordunif2}, for all $0<\eps<|\rho|$,
%
%e5.1 #&#
\begin{eqnarray}\label{eq:Delta1}
\Delta_1(x) & :=& \frac{U(d_nx)-b}a -\frac{(xd_nk/n)^\gamma-1}\gamma
\nonumber
\\[-8pt]\\[-8pt]
& =& \mathrm{O} \bigl(A(n/k) (xd_nk/n)^{\gamma+\rho\pm\eps} \bigr) =
\mathrm{o} \bigl(k^{-1/2} w_n(\gamma) (xd_nk/n)^\gamma
\bigr) \nonumber
\end{eqnarray}
uniformly for all $x\ge\lambda_n$, where in the last step we have
used condition (i).
Now one can conclude that $U(d_nx)\in\hat T_n((0,\infty))$ for all
$x\in[\lambda_n,\tau_n]$ with probability tending to 1. For
example, if $\gamma>0$, then we have to show that $U(d_nx)>\hat
b-\hat a/\hat\gamma$ for all $x\ge\lambda_n$ or, equivalently,
(using (M3)) that $\Delta_1(\lambda_n) $ is larger than
\[
\frac{\hat b-b}a-\frac{\hat
a}{a\hat\gamma} - \frac{(\lambda_nd_nk/n)^\gamma-1}\gamma= -
\frac{1}\gamma \biggl(\frac{d_nk}n\lambda_n
\biggr)^\gamma+\mathrm{O}\bigl(k^{-1/2}\bigr)
\]
which follows immediately from \eqref{eq:Delta1}, (S1) and (ii).

Hence
\begin{eqnarray*}
&&T_{n}^\leftarrow\circ\hat T_{n} \circ \hat
T_n^{(c)\leftarrow}\circ U(d_n x)\\
&&\quad  = \biggl[1+\frac\gamma a
\biggl(\hat a \frac{  (c_n^{-1}  (
1+\hat\gamma(\vfrac{U(d_nx)-\hat b}{\hat
a}) )^{1/\hat\gamma} )^{\hat\gamma}-1}{\hat\gamma} +\hat b-b \biggr) \biggr]^{1/\gamma}
 =:
\tilde H(x)
\end{eqnarray*}
if the expression in brackets is strictly positive, which
will indeed follow from the calculations below.

Now direct calculations show that
%
%e5.2 #&#
\begin{equation}
\label{eq:tildeHrep} \tilde H(x) = \biggl[1+\gamma \biggl( c_n^{-\hat\gamma}
\frac{U(d_nx)-b}a + \frac{c_n^{-\hat\gamma} -1}{\hat\gamma} \biggl(\frac{\hat a}a -
\frac{\hat b-b}a\hat\gamma \biggr) \biggr) \biggr]^{1/\gamma}.
\end{equation}
By assumption (M3)
%
%e5.3 #&#
\begin{equation}
\label{eq:Delta2} \Delta_2 := \frac{\hat a}a -\frac{\hat b-b}a
\hat\gamma-1 = k^{-1/2}\bigl(\alpha-\gamma\beta+\mathrm{o}_P(1)
\bigr).
\end{equation}

If $\gamma> 0$, then the Taylor expansion
\[
c_n^{-\hat\gamma/\gamma} = c_n^{-1}
\biggl(1-k^{-1/2}\frac\Gamma\gamma\log c_n+
\mathrm{o}_P\bigl(k^{-1/2}\log c_n\bigr) \biggr)
\]
together with \eqref{eq:Delta1}, \eqref{eq:Delta2} and (S1) implies
\begin{eqnarray*}
\tilde H(x) & = & \biggl[ c_n^{-\hat\gamma} \biggl( \biggl(
\frac{d_n k}{n} x \biggr)^\gamma-1 + \gamma\Delta_1(x) +
\frac{\gamma}{\hat\gamma}(1+\Delta_2) \biggr) +1-\frac{\gamma}{\hat\gamma} -
\frac{\gamma}{\hat\gamma}\Delta_2 \biggr]^{1/\gamma}
\\
& = & c_n^{-\hat\gamma/\gamma} \frac{d_nk}nx
\\
& & {}\times \biggl[1+\mathrm{O}_P \biggl( \bigl( \bigl|
\Delta_1(x)\bigr|+k^{-1/2} \bigr) \biggl(\frac{d_nk}n x
\biggr)^{-\gamma} \biggr)+ \mathrm{O}_P \biggl(k^{-1/2}
c_n^{\hat\gamma} \biggl(\frac{d_nk}n x \biggr)^{-\gamma}
\biggr) \biggr]^{1/\gamma}
\\
& = & \frac{d_n}{e_n} x \biggl(1-k^{-1/2}\frac\Gamma\gamma\log
c_n+\mathrm{o}_P\bigl(k^{-1/2}\log
c_n\bigr) \biggr)
\\
& & {}\times \biggl[ 1+\mathrm{o}_P\bigl(k^{-1/2}\log
c_n\bigr) +\mathrm{O}_P \biggl(k^{-1/2} \biggl(
\frac{d_n}{e_n} x \biggr)^{-\gamma} \biggr) \biggr],
\end{eqnarray*}
from which the assertion follows readily.

If $\gamma<0$, then similar arguments prove
\begin{eqnarray*}
\tilde H(x) & = & \frac{d_n}{e_n} x \bigl(1+\mathrm{O}_P
\bigl(k^{-1/2}\log c_n\bigr) \bigr)
\\
& &{} \times \biggl[ 1+ k^{-1/2} \frac{1}\gamma \biggl(\alpha-
\gamma\beta-\frac \Gamma\gamma+\mathrm{o}_P(1) \biggr) \biggl(
\frac{d_nk}n x \biggr)^{-\gamma}+\mathrm{o}_P
\bigl(k^{-1/2} w_n(\gamma)\bigr) \biggr]
\end{eqnarray*}
and hence the assertion, because the assumption (iii) ensures that
$\log c_n=\mathrm{o}(w_n(\gamma))$.

Finally, for $\gamma=0$, the Taylor expansion
\[
c_n^{-\hat\gamma} = 1-\hat\gamma\log c_n +
\tfrac{1}2\hat\gamma^2\log^2c_n +
\mathrm{O}_P\bigl(\hat\gamma^3\log^3
c_n\bigr)
\]
yields
\begin{eqnarray*}
\tilde H(x) & = & \exp \biggl[ \bigl(1-\hat\gamma\log c_n +
\mathrm{O}_P\bigl(k^{-1}\log^2 c_n
\bigr) \bigr) \biggl(\log \biggl(\frac{d_nk}nx \biggr)+\Delta_1(x)
\biggr)
\\
& & \hphantom{\exp \biggl[}{ } + \biggl(-\log c_n+\frac{1}2\hat\gamma
\log^2 c_n+\mathrm{O}_P\bigl(k^{-1}
\log^3 c_n\bigr) \biggr) (1+\Delta_2) \biggr]
\\
& = & \frac{d_nk}{c_nn}x \exp \biggl[ -\hat\gamma\log c_n \biggl(
\log c_n+\log \biggl(\frac{d_n}{e_n}x \biggr) \biggr)+
\mathrm{o}_P\bigl(k^{-1/2}\log^2 c_n
\bigr)+\frac{1}2\hat\gamma\log^2 c_n \biggr]
\\
& = & \frac{d_n}{e_n} x \biggl[1-\frac{1}2 \bigl(\Gamma+
\mathrm{o}_P(1)\bigr)k^{-1/2}\log^2
c_n + \mathrm{O}_P\bigl(k^{-1/2}\log
c_n \log x\bigr) \biggr],
\end{eqnarray*}
which concludes the proof.
\end{pf}

In what follows we denote by $\lambda_{n,1}\searrow x_l$,
$\lambda_{n,2}\searrow q(\infty):=\lim_{x\to\infty} q(x)$ and
$\tau_{n,i}\uparrow\infty$, $i\in\{1,2\}$, sequences which satisfy
the conditions of Lemma~\ref{lemma:marginalapprox}. (These sequences
will be specified in the proof of Corollary~\ref{cor:termIVVapprox}.) Note that in particular constant sequences
$\lambda_{n,i}, \tau_{n,i}\in(0,\infty)$ satisfy the conditions of
Lemma~\ref{lemma:marginalapprox}, provided
%
%e5.4 #&#
\begin{equation}
\label{Aicond} A_i(n/k) c_n^{\rho_i+\eps} =
\mathrm{o} \bigl(k_n^{-1/2} w_n(
\gamma_i) \bigr) \qquad \mbox{for } i\in\{1,2\} \mbox{ and some } \eps>0
\end{equation}
and (S1) holds. Therefore, we may and will choose
%
%e5.5 #&#
%e5.6 #&#
\begin{eqnarray}
\label{lambdatauchoice} %
\begin{array}{rcl@{\qquad}l}
\lambda_{n,1} &=&x_l &\mbox{if
} x_l:=\inf\bigl\{ x\ge0\mid q(x)<\infty\bigr\} >0,
\\
\lambda_{n,2} &=&q(\infty) &\mbox{if } q(\infty)>0.
\end{array} %
\end{eqnarray}

We want to apply the approximations just established to points
$(x,y)$ on the boundary of $S$. To ensure that
$x\in[\lambda_{n,1},\tau_{n,1}]$ and
$y\in[\lambda_{n,2},\tau_{n,2}]$, we consider a subset $S_n^*$ of
$S$ that is bounded away from the coordinate axes.
More precisely, we define
\[
S_n^* := S\cap \bigl(\bigl[u_n^*,\infty\bigr)
\times\bigl[v_n^*,\infty\bigr) \bigr)
\]
with
\[
u_n^* := \lambda_{n,1} \vee q^\leftarrow(
\tau_{n,2}),\qquad  v_n^* := \lambda_{n,2} \vee q (
\tau_{n,1}).
\]
The following lemma implies that the $\nu$-measure of the set
$S\setminus S_n^*$ is asymptotically negligible.

%le5.2 #&#
\begin{lemma} \label{lemma:Snstarapprox}
\[
\nu(S)-\nu\bigl(S_n^*\bigr) = \mathrm{O} \biggl( \frac{\lambda_{n,1}-x_l}{q^2(\lambda
_{n,1}-)}
+ \frac{q(\tau_{n,1})-q(\infty)}{\tau_{n,1}^2}+ \frac{\lambda_{n,2}-q(\infty)}{(q^\leftarrow(\lambda_{n,2}))^2} +\frac{q^\leftarrow(\tau_{n,2})-x_l}{\tau_{n,2}^2} \biggr)
\]
with $q(x-):=\lim_{t\uparrow x} q(t)$.
\end{lemma}

\begin{pf}
First, note that $S\subset[x_l,\infty)\times[q(\infty),\infty)$
implies
\begin{eqnarray*}
\nu(S)-\nu \bigl(S\cap\bigl(\bigl[0,u_n^*\bigr)\times[0,\infty)\bigr) \bigr)
&\le& \nu \bigl([x_l,\lambda_{n,1})\times\bigl[q(
\lambda_{n,1}-),\infty\bigr) \bigr)\\
&&{} + \nu \bigl(\bigl[x_l,q^\leftarrow(
\tau_{n,2})\bigr)\times[\tau_{n,2},\infty ) \bigr).
\end{eqnarray*}
The spectral density $\varphi$ is assumed continuous and hence it
is bounded.
From \eqref{eq:etanurel}, we conclude
that for arbitrary $0\le u_0\le u_1$ and $v_0>0$
\[
\nu \bigl([u_0,u_1)\times[v_0,\infty)
\bigr) = \mathrm{O} \biggl( \int_{u_0}^{u_1} \int
_{v_0}^\infty\bigl(u^2+v^2
\bigr)^{-3/2} \,\mathrm{d}v \,\mathrm{d}u \biggr) = \mathrm{O} \biggl(\frac{u_1-u_0}{v_0^2} \biggr)
\]
and thus
\[
\nu(S)-\nu \bigl(S\cap\bigl(\bigl[0,u_n^*\bigr)\times[0,\infty)\bigr) \bigr) =
\mathrm{O} \biggl( \frac{\lambda_{n,1}-x_l}{q^2(\lambda_{n,1}-)} + \frac{q^\leftarrow(\tau_{n,2})-x_l}{\tau_{n,2}^2} \biggr).
\]
Likewise, one can show that
\[
\nu \bigl(S\cap\bigl(\bigl[0,u_n^*\bigr)\times[0,\infty)\bigr) \bigr)-\nu
\bigl(S_n^*\bigr) = \mathrm{O} \biggl( \frac{q(\tau_{n,1})-q(\infty)}{\tau_{n,1}^2}+
\frac{\lambda_{n,2}-q(\infty)}{(q^\leftarrow(\lambda
_{n,2}))^2} \biggr).
\]
A combination of these two bounds yields the assertion.
\end{pf}

On the set $S_n^*$ we can now use the approximation from Lemma~\ref{lemma:marginalapprox} to first examine the influence of the
transformation $H_{n,2}$ of the second coordinate on the
$\nu$-measure of $S_n^*$. In a second step, we then similarly
determine how the $\nu$-measure of this transformed set is altered
by the transformation $H_{n,1}$ of the first coordinate. Hereby, note
that by Lemma~\ref{lemma:marginalapprox} the marginal
transformations are invertible with probability tending to 1.
The sets which are relevant for the comparison after the first marginal transformation are depicted in Figure~\ref{sketch}.

%le5.3 #&#
\begin{lemma} \label{lemma:doubleint}
Let $H_n(x,y) := (H_{n,1}(x),H_{n,2}(y)) := \frac{e_n}{d_n}
T_{n}^\leftarrow\circ\hat T_{n} \circ
\hat T_n^{(c)\leftarrow}\circ U(d_nx,d_ny)$. Suppose that the
conditions \textup{(D2)} and \textup{(Q1)} are met.

Then one has with $q_n(u) := q(u)\vee v_n^*$ and $\tilde
q_n^\leftarrow(v):= q^\leftarrow(H_{n,2}^\leftarrow(v))\vee u_n^*$
%
%e5.7 #&#
\begin{eqnarray}
\label{eq:diffmeasapprox} &&\biggl| \nu\bigl(H_n\bigl(S_n^*\bigr)\bigr) -
\nu\bigl(S_n^*\bigr) +\int_{u_n^*}^\infty
\bigl(H_{n,2}\bigl(q_n(u)\bigr)-q_n(u) \bigr)
\eta\bigl(u,q_n(u)\bigr)\, \mathrm{d}u
\nonumber
\\
& &\hphantom{\biggl|}  { } + \int_{H_{n,2}(v_n^*)}^\infty \bigl(H_{n,1}
\bigl(\tilde q_n^\leftarrow(v)\bigr)- \tilde q_n^\leftarrow(v)
\bigr) \eta\bigl(\tilde q_n^\leftarrow(v),v\bigr) \,\mathrm{d}v \biggr|
\nonumber
\\[-8pt]\\[-8pt]
& &\quad =  \mathrm{o} \biggl(\int_{u_n^*}^\infty
\bigl|H_{n,2}\bigl(q_n(u)\bigr)-q_n(u) \bigr| \eta
\bigl(u,q_n(u)\bigr)\,\mathrm{d}u
\nonumber
\\
& &\hphantom{\quad =  \mathrm{o} \biggl(}{} + \int_{H_{n,2}(v_n^*)}^\infty \bigl|H_{n,1}\bigl(
\tilde q_n^\leftarrow (v)\bigr)- \tilde q_n^\leftarrow(v)
\bigr| \eta\bigl(\tilde q_n^\leftarrow(v),v\bigr) \,\mathrm{d}v \biggr)
\nonumber
\end{eqnarray}
with probability tending to $1$.
\end{lemma}

\begin{pf}
According to the proof of Lemma~\ref{lemma:marginalapprox}, for all
$\delta\in(0,1)$, on the set
$[\lambda_{n,i}(1-\delta),\tau_{n,i}(1+\delta)]$ the transformation
$H_{n,i}$ is continuous and strictly increasing and
$H_{n,i}(x)=x(1+\mathrm{o}(1))$ with probability tending to 1.
%
%f7 #&#
\begin{figure}

\includegraphics{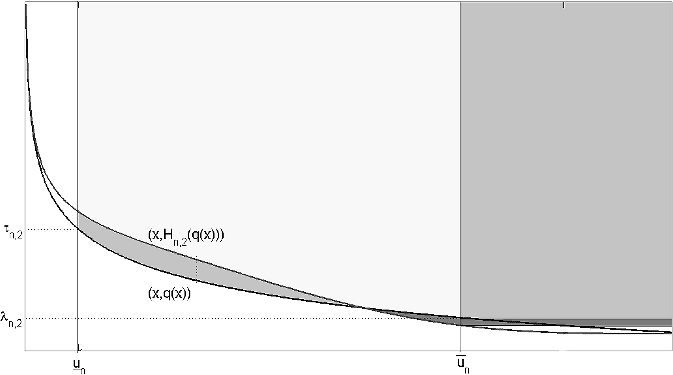}

\caption{The light and mid grey regions show the approximation
$S_n^*$ of the set $S$,
the mid and the dark grey regions the symmetric difference
between $\{(x,H_{n,2}(y))\mid(x,y)\in S_n^*\}$ and $S_n^*$, where the
dark grey region
is counted with a positive sign, the mid grey
region with a negative sign. (Here it is assumed that
$u_n^*=q^\leftarrow(\tau_{n,2})$.)} \label{sketch}
\end{figure}

We first quantify the influence of the
transformation of the second coordinate. Note that
\begin{eqnarray*}
\nu\bigl(S_n^*\bigr) & = & \int_{u_n^*}^\infty
\int_{q_n(u)}^\infty\eta(u,v) \,\mathrm{d}v \,\mathrm{d}u,
\\
\nu\bigl\{ \bigl(x,H_{n,2}(y)\bigr) \mid(x,y)\in S_n^*
\bigr\} & = & \int_{u_n^*}^\infty \int
_{H_{n,2}(q_n(u))}^\infty\eta(u,v) \,\mathrm{d}v \,\mathrm{d}u
\end{eqnarray*}
and hence
%
%e5.8 #&#
\begin{equation}
\label{eq:diffbound1} \nu\bigl\{ \bigl(x,H_{n,2}(y)\bigr) \mid(x,y)\in
S_n^*\bigr\}-\nu\bigl(S_n^*\bigr) =- \int
_{u_n^*}^\infty\int_{q_n(u)}^{H_{n,2}(q_n(u))}
\eta(u,v)\,\mathrm{d}v \,\mathrm{d}u.
\end{equation}

The inner integral equals
%
%e5.9 #&#
\begin{eqnarray}
\label{eq:intsimple}&& \int_{q_n(u)}^{H_{n,2}(q_n(u))} \eta(u,v) \,\mathrm{d}v
\nonumber
\\
&&\quad  =  \bigl(H_{n,2}\bigl(q_n(u)\bigr)-q_n(u)
\bigr) \eta\bigl(u,q_n(u)\bigr)
\\
& & \qquad {}+ \int_1^{H_{n,2}(q_n(u))/q_n(u)} \biggl(\frac{\eta(u,q_n(u)w)}{\eta(u,q_n(u))}-1
\biggr) \,\mathrm{d}w \eta\bigl(u,q_n(u)\bigr) q_n(u).
\nonumber
\end{eqnarray}
By the assumptions and Lemma~\ref{lemma:marginalapprox},
$H_{n,2}(q_n(u))/q_n(u)\to1$ uniformly for $u\in(u_n^*,\infty)$ as
$q_n(u)\in[\lambda_{n,2},\tau_{n,2}]$ for $u>u_n^*$.

Next note that for all $0<\delta<\uppi/4$
\[
\sup_{0< t\le\tan\delta} \biggl|\frac{\arctan(tw)}{\arctan
t}-1 \biggr| \le\sup
_{0<t\le\tan\delta} \frac{t|w-1|}{\arctan t} \le\bigl(1+\tan^2\delta
\bigr)|w-1|\to0
\]
as $w\to1$, and likewise by symmetry
\[
\sup_{t\ge\tan(\uppi/2-\delta)} \biggl|\frac{\uppi/2-\arctan
(tw)}{\uppi/2-\arctan t}-1 \biggr| =\sup_{0<t\le\tan\delta}
\biggl|\frac{\arctan(t/w)}{\arctan t}-1 \biggr| \to0.
\]
Therefore, by condition \eqref{eq:phivar}, to each $\eps>0$ there
exists $\delta>0$ such that for all $t\in(0,\tan\delta]\cup[\tan
(\uppi/2-\delta),\infty)$
\[
\biggl|\frac{\varphi(\arctan(tw))}{\varphi(\arctan t)}-1 \biggr|<\eps .
\]
Moreover, by the uniform continuity of $\varphi\circ\arctan$ on
$[\tan\delta,\tan(\uppi/2-\delta)]$
\begin{eqnarray*}
&& \sup_{\tan\delta\le t\le\tan(\uppi/2-\delta)}\biggl |\frac
{\varphi(\arctan(tw))}{\varphi(\arctan t)}-1 \biggr|
\\
&&\quad  \le \frac{\sup_{\tan\delta\le t\le\tan(\uppi/2-\delta)} |\varphi
(\arctan(tw))-\varphi(\arctan t)|}{\inf_{\delta\le u\le\uppi
/2-\delta}\varphi(u)} \to0
\end{eqnarray*}
as $w\to1$. Using \eqref{eq:etanurel} and
\[
\frac{1+t^2}{1+t^2w^2} = 1+ \frac{1-w^2}{t^{-2}+w^2} \to1
\]
as $w\to1$ uniformly for $t>0$, we conclude that
\[
\frac{\eta(u,vw)}{\eta(u,v)} = \biggl( \frac{1+(v/u)^2}{1+(v/u)^2w^2} \biggr)^{3/2}
\frac{\varphi
(\arctan
\vfrac{vw}u )}{\varphi (\arctan\sfrac{v}u )} \to1
\]
as $w\to1$ uniformly for $u,v>0$.
Thus,
\[
\int_1^{H_{n,2}(q_n(u))/q_n(u)} \biggl(\frac{\eta(u,q_n(u)w)}{\eta
(u,q_n(u))}-1
\biggr)\, \mathrm{d}w = \mathrm{o} \bigl(H_{n,2}\bigl(q_n(u)
\bigr)/q_n(u)-1 \bigr)
\]
which, combined with \eqref{eq:diffbound1} and \eqref{eq:intsimple},
yields
%
%e5.10 #&#
\begin{eqnarray}\label{eq:doubleintapprox1}
&&\nu\bigl\{ \bigl(x,H_{n,2}(y)\bigr) \mid(x,y)\in S_n^*
\bigr\}-\nu\bigl(S_n^*\bigr) +\int_{u_n^*}^\infty
\bigl(H_{n,2}\bigl(q_n(u)\bigr)-q_n(u) \bigr)
\eta\bigl(u,q_n(u)\bigr) \,\mathrm{d}u
\nonumber
\\[-8pt]\\[-8pt]
&&\quad  =  \mathrm{o} \biggl(\int_{u_n^*}^\infty
\bigl|H_{n,2}\bigl(q_n(u)\bigr)-q_n(u) \bigr| \eta
\bigl(u,q_n(u)\bigr) \,\mathrm{d}u \biggr). \nonumber
\end{eqnarray}

One can derive an analogous approximation of the difference between
$\nu\{ (x,H_{n,2}(y)) \mid(x,y)\in S\}$ and $\nu\{
(H_{n,1}(x),H_{n,2}(y)) \mid(x,y)\in S\}$ by similar arguments if
one interchanges the order of integration:
%
%e5.11 #&#
\begin{eqnarray}\label{eq:secapprox}
&&\biggl| \nu\bigl(H_n\bigl(S_n^*\bigr)\bigr) - \nu\bigl\{
\bigl(x,H_{n,2}(y)\bigr) \mid(x,y)\in S_n^*\bigr\}
\nonumber
\\
& &\hphantom{\biggl| } {} + \int_{H_{n,2}(v_n^*)}^\infty \bigl(H_{n,1}
\bigl(\tilde q_n^\leftarrow(v)\bigr)- \tilde q_n^\leftarrow(v)
\bigr) \eta\bigl(\tilde q_n^\leftarrow(v),v\bigr) \,\mathrm{d}v \biggr|
\\
&&\quad  =  \mathrm{o} \biggl(\int_{H_{n,2}(v_n^*)}^\infty
\bigl|H_{n,1}\bigl(\tilde q_n^\leftarrow(v)\bigr)- \tilde
q_n^\leftarrow(v) \bigr| \eta\bigl(\tilde q_n^\leftarrow(v),v
\bigr) \,\mathrm{d}v \biggr). \nonumber
\end{eqnarray}
Summing up \eqref{eq:doubleintapprox1} and \eqref{eq:secapprox}, we
arrive at the assertion.
\end{pf}

In the next lemma, we calculate the limits of the integrals arising
in Lemma~\ref{lemma:doubleint} using the approximation established
in Lemma~\ref{lemma:marginalapprox}.

%le5.4 #&#
\begin{lemma} \label{lemma:singleintapprox}
Suppose that the conditions of Lemma~\ref{lemma:doubleint} and, in
addition, the following conditions are
fulfilled for some $x_0\in(x_l,q^\leftarrow(q(\infty))), y_0\in
(q(\infty),q(x_l))$:
%
%e5.12 #&#
%e5.13 #&#
\begin{eqnarray}
\label
{Q2.1}\int_{y_0}^{\infty} \bigl(q^\leftarrow(v)
\bigr)^{1-\gamma_1} v^{-3} \,\mathrm{d}v <\infty\quad \mbox{or}\quad
\lambda_{n,1}^{1-\gamma_1}=\mathrm{o}(\log c_n),
\\
\label{Q2.2}\int_{x_0}^{\infty} \bigl(q(u)\bigr)^{1-\gamma_2}
u^{-3} \,\mathrm{d}u<\infty \quad \mbox{or}\quad  \lambda_{n,2}^{1-\gamma_2}=
\mathrm{o}(\log c_n).
\end{eqnarray}

Then the following approximations hold true:
\begin{enumerate}[(ii)]
\item[(i)]
\begin{eqnarray*}
&&\frac{k^{1/2}}{w_n(\gamma_2)}\int_{u_n^*}^\infty
\bigl(H_{n,2}\bigl(q_n(u)\bigr)-q_n(u) \bigr)
\eta\bigl(u,q_n(u)\bigr)\, \mathrm{d}u
\\
& &\quad \to \lleft\{ %
\begin{array} {l@{\qquad  }l} - \displaystyle \frac{\Gamma_2}{\gamma_2} \displaystyle \int
_{x_l}^\infty q(u)\eta\bigl(u,q(u)\bigr) \,\mathrm{d}u, &
\gamma_2>0,
\\\noalign{\vspace*{-1pt}}
\biggl(\displaystyle \frac{\alpha_2}{\gamma_2}-\beta_2-\displaystyle \frac{\Gamma_2}{\gamma
_2^2} \biggr) \int
_{x_l}^\infty \bigl(q(u)\bigr)^{1-\gamma_2} \eta
\bigl(u,q(u)\bigr) \,\mathrm{d}u, & \gamma_2<0,
\\\noalign{\vspace*{-1pt}}
- \Gamma_2\displaystyle \int_{x_l}^\infty q(u)\eta
\bigl(u,q(u)\bigr) \,\mathrm{d}u, & \gamma_2=0. \end{array} %
\rright.
\end{eqnarray*}
Moreover,
\[
\int_{u_n^*}^\infty \bigl|H_{n,2}
\bigl(q_n(u)\bigr)-q_n(u) \bigr| \eta\bigl(u,q_n(u)
\bigr) \,\mathrm{d}u = \mathrm{O}\bigl(k^{-1/2}w_n(\gamma_2)
\bigr).
\]
\item[(ii)]
\begin{eqnarray*}
&&\frac{k^{1/2}}{w_n(\gamma_1)}\int_{H_{n,2}(v_n^*)}^\infty
\bigl(H_{n,1}\bigl(\tilde q_n^\leftarrow(v)\bigr)-
\tilde q_n^\leftarrow(v) \bigr) \eta\bigl(\tilde
q_n^\leftarrow(v),v\bigr) \,\mathrm{d}v
\\
&&\quad  \to \lleft\{ %
\begin{array} {l@{\qquad  }l} - \displaystyle \frac{\Gamma_1}{\gamma_1} \displaystyle \int
_{q(\infty)}^\infty q^\leftarrow (v)\eta
\bigl(q^\leftarrow(v),v\bigr) \,\mathrm{d}v, & \gamma_1>0,
\\
\biggl(\displaystyle \frac{\alpha_1}{\gamma_1}-\beta_1-\displaystyle \frac{\Gamma_1}{\gamma
_1^2} \biggr) \displaystyle \int
_{q(\infty)}^\infty \bigl(q^\leftarrow(v)
\bigr)^{1-\gamma_1} \eta\bigl(q^\leftarrow(v),v\bigr) \,\mathrm{d}v, & \gamma
_1<0,
\\
- \Gamma_1 \displaystyle \int_{q(\infty)}^\infty
q^\leftarrow(v)\eta \bigl(q^\leftarrow(v),v\bigr)\, \mathrm{d}v, &
\gamma_1=0. \end{array} %
\rright.
\end{eqnarray*}
Furthermore,
\[
\int_{H_{n,2}(v_n^*)}^\infty \bigl|H_{n,1}\bigl(\tilde
q_n^\leftarrow(v)\bigr)- \tilde q_n^\leftarrow(v)
\bigr| \eta\bigl(\tilde q_n^\leftarrow(v),v\bigr) \,\mathrm{d}v = \mathrm{O}
\bigl(k^{-1/2}w_n(\gamma_1)\bigr).
\]
\end{enumerate}
\end{lemma}

\begin{pf}
ad (i): Because the spectral density $\varphi$ is
bounded, there exists a constant $K>0$ such that
%
%e5.14 #&#
\begin{equation}
\label{eq:etabound} \eta\bigl(u,q(u)\bigr) \le K \bigl(u^2+\bigl(q(u)
\bigr)^2 \bigr)^{-3/2} \le K \bigl(u^{-3}\wedge
\bigl(q(u)\bigr)^{-3} \bigr)\qquad  \forall u>0.
\end{equation}
Hence, $q_n(u)\eta(u,q_n(u))\le K(q(u))^{-2}$ for $u\in[x_l,x_0]$
and $n$ sufficiently large and $q_n(u)$ $\eta(u,q_n(u))\le K
q(x_0)u^{-3}$ for $u>x_0$. Therefore,
%
%e5.15 #&#
\begin{equation}
\label{eq:etaintbound} \lim_{n\to\infty} \int_{u_n^*}^\infty
q_n(u)\eta\bigl(u,q_n(u)\bigr) \,\mathrm{d}u =\int
_{x_l}^\infty q(u)\eta\bigl(u,q(u)\bigr) \,\mathrm{d}u<\infty
\end{equation}
by the dominated convergence
theorem and $u_n^*\downarrow x_l$.

Now, we distinguish three cases.

If $ \gamma_2>0$, then by Lemma~\ref{lemma:marginalapprox} and $d_n\asymp e_n$
\begin{eqnarray*}
&&\hspace*{-5pt}\frac{k^{1/2}}{\log c_n}\int_{u_n^*}^\infty
\bigl(H_{n,2}\bigl(q_n(u)\bigr)-q_n(u) \bigr)
\eta\bigl(u,q_n(u)\bigr) \,\mathrm{d}u
\\
&&\hspace*{-5pt}\quad  =  - \biggl(\frac{\Gamma_2}{\gamma_2}+\mathrm{o}_P(1) \biggr) \int
_{u_n^*}^\infty q_n(u)\eta
\bigl(u,q_n(u)\bigr) \,\mathrm{d}u
 + \mathrm{O}_P \biggl(\frac{1}{\log c_n} \int
_{u_n^*}^\infty \bigl(q_n(u)
\bigr)^{1-\gamma_2} \eta\bigl(u,q_n(u)\bigr) \,\mathrm{d}u \biggr).
\end{eqnarray*}
Because of \eqref{eq:etabound} and \eqref{Q2.2}
%
%e5.16 #&#
\begin{eqnarray}\label{eq:etaintbound2}
&&\int_{u_n^*}^\infty\bigl(q_n(u)
\bigr)^{1-\gamma_2}\eta\bigl(u,q_n(u)\bigr) \,\mathrm{d}u
\nonumber
\\
& &\quad \le K \bigl(q(x_0)\bigr)^{-2-\gamma_2}(x_0-x_l)
+ K\int_{x_0}^\infty\bigl(q(u)\vee
\lambda_{n,2}\bigr)^{1-\gamma_2}u^{-3} \,\mathrm{d}u
\\
&&\quad  =  \mathrm{o}(\log c_n).\nonumber
\end{eqnarray}
Hence, in view of \eqref{eq:etaintbound}, we have
\begin{eqnarray*}
&&\int_{u_n^*}^\infty \bigl(H_{n,2}
\bigl(q_n(u)\bigr)-q_n(u) \bigr)\eta\bigl(u,q_n(u)
\bigr) \,\mathrm{d}u
\\
&&\quad  =  -k^{-1/2} \log c_n \frac{\Gamma_2}{\gamma_2} \int
_{x_l}^\infty q(u)\eta\bigl(u,q(u)\bigr) \,\mathrm{d}u +
\mathrm{o}_P\bigl(k^{-1/2}\log c_n\bigr).
\end{eqnarray*}

If $ \gamma_2<0$, then the assertion follows similarly
from Lemma~\ref{lemma:marginalapprox} and \eqref{eq:etabound}.

Finally, in the case $ \gamma_2=0$
%
%e5.17 #&#
\begin{eqnarray}\label{eq:etaintbound3}
&&\int_{x_l}^\infty q_n(u)\bigl|\log
q_n(u)\bigr|\eta\bigl(u,q_n(u)\bigr) \,\mathrm{d}u
\nonumber
\\[-8pt]\\[-8pt]
& &\quad \le K \sup_{x\le x_0} \frac{|\log q(x)|}{(q(x))^2} + K \sup
_{x\ge x_0} q(x)\bigl|\log q(x)\bigr|\int_{x_0}^\infty
u^{-3} \,\mathrm{d}u <\infty. \nonumber
\end{eqnarray}
Hence, similarly as in the first case, we may conclude the assertion
from Lemma~\ref{lemma:marginalapprox}.

ad (ii): The second assertion can be proved in a very similar
fashion using $q(H_{n,2}(v_n^*))\to q(\infty)$ and the fact that
$\tilde q_n^\leftarrow(u)\to q^\leftarrow(u)$ for Lebesgue-almost
all $u>q(\infty)$, because of Lemma~\ref{lemma:marginalapprox} and
the Lebesgue-almost surely continuity of $q^\leftarrow$. For that
reason, we only give the analog to the bound
\eqref{eq:etaintbound2} for the integral
under consideration in the case $\gamma_1>0$.

For $y_0\in(q(\infty),q(x_l))$ and all sufficiently large $n$, we
have
\[
\int_{H_{n,2}(v_n^*)}^{y_0} \bigl(\tilde q_n^\leftarrow(v)
\bigr)^{1-\gamma_1} \eta \bigl(\tilde q_n^\leftarrow(v),v \bigr)
\,\mathrm{d}v \le K\bigl(\tilde q_n^\leftarrow(y_0)
\bigr)^{-2-\gamma_1}\bigl(y_0-q(\infty)\bigr)=\mathrm{O}(1).
\]
If $\gamma_1\le1$, then
\[
\int_{y_0}^\infty \bigl(\tilde q_n^\leftarrow(v)
\bigr)^{1-\gamma_1} \eta \bigl(\tilde q_n^\leftarrow(v),v \bigr)
\,\mathrm{d}v \le K \bigl(\tilde q_n^\leftarrow(y_0)
\bigr)^{1-\gamma_1}\int_{y_0}^\infty
v^{-3} \,\mathrm{d}v = \mathrm{O}(1).
\]
Finally, if $\gamma_1>1$, then by the monotonicity of $q^\leftarrow$
and the asymptotic behavior of $H_{n,2}$ we have for all $\delta>0$
and sufficiently large $n$
\begin{eqnarray*}
&&\int_{y_0}^\infty \bigl(\tilde q_n^\leftarrow(v)
\bigr)^{1-\gamma_1} \eta \bigl(\tilde q_n^\leftarrow(v),v \bigr)
\,\mathrm{d}v
\\
&&\quad  \le K \int_{y_0}^\infty \bigl( \bigl(
q^\leftarrow\bigl(v(1+\delta)\bigr)\bigr)^{1-\gamma_1} \wedge
\lambda_{n,1}^{1-\gamma_1} \bigr) v^{-3} \,\mathrm{d}v
\\
& &\quad =  \mathrm{O} \biggl(\int_{y_0(1+\delta)}^\infty
\bigl(q^\leftarrow(v)\bigr)^{1-\gamma
_1}v^{-3} \,\mathrm{d}v \wedge
\lambda_{n,1}^{1-\gamma_1} \biggr) = \mathrm{o}(\log c_n)
\end{eqnarray*}
by condition \eqref{Q2.1}.
\end{pf}

The following result gives sufficient conditions such that the
difference between the $\nu$-measure of $S$ and of the truncated set
after the marginal transformations (i.e., $d_n(\mathit{IV}+V)$ in
\eqref{eq:esterror}) can be approximated by the limiting terms in
Lemma~\ref{lemma:singleintapprox}. For the sake of simplicity, we
assume that $d_n$ and $e_n$ are of the same order, but it is not
difficult to prove similar results under weaker conditions on
$d_n/e_n$. Moreover, one can weaken the condition (S2) and the
assumptions (Q2) could be replaced with rather strong conditions on
the rate at which $k$ tends to $\infty$.

%
%co5.5 #&#
\begin{corollary} \label{cor:termIVVapprox}
If the conditions \textup{(M1)--(M3), (D2), (Q1), (Q2)} and \textup{(S1)--(S3)} are
fulfilled,
then
%
%e5.18 #&#
%e5.19 #&#
\begin{eqnarray}\label{eq:nuHnapprox}
&&\hspace*{-25pt}\nu\bigl(H_n\bigl(S_n^*\bigr)\bigr) - \nu(S)
\nonumber
\\
&&\hspace*{-25pt}\quad  =  k^{-1/2} w_n(\gamma_1)\lleft\{
\begin{array} {l@{\qquad  }l} -\displaystyle  \frac{\Gamma_1}{\gamma_1} \displaystyle \int_{q(\infty)}^\infty
q^\leftarrow (v)\eta\bigl(q^\leftarrow(v),v\bigr) \,\mathrm{d}v, &
\gamma_1>0,
\\
\biggl(\displaystyle \frac{\alpha_1}{\gamma_1}-\beta_1-\displaystyle \frac{\Gamma_1}{\gamma
_1^2} \biggr) \displaystyle \int
_{q(\infty)}^\infty \bigl(q^\leftarrow(v)
\bigr)^{1-\gamma_1} \eta\bigl(q^\leftarrow(v),v\bigr) \,\mathrm{d}v, & \gamma
_1<0,
\\
- \Gamma_1 \displaystyle \int_{q(\infty)}^\infty
q^\leftarrow(v)\eta \bigl(q^\leftarrow(v),v\bigr) \,\mathrm{d}v, &
\gamma_1=0, \end{array} %
\rright.\nonumber
\\[-8pt]\\[-8pt]
& &\hspace*{-25pt}\qquad  { } +k^{-1/2} w_n(\gamma_2) \lleft\{
\begin{array} {l@{\qquad  }l} - \displaystyle \frac{\Gamma_2}{\gamma_2} \displaystyle \int_{x_l}^\infty
q(u)\eta\bigl(u,q(u)\bigr) \,\mathrm{d}u, & \gamma_2>0 ,
\\
\biggl(\displaystyle \frac{\alpha_2}{\gamma_2}-\beta_2-\displaystyle \frac{\Gamma_2}{\gamma
_2^2} \biggr) \displaystyle \int
_{x_l}^\infty \bigl(q(u)\bigr)^{1-\gamma_2} \eta
\bigl(u,q(u)\bigr) \,\mathrm{d}u, & \gamma_2<0,
\\
- \Gamma_2\displaystyle \int_{x_l}^\infty q(u)\eta
\bigl(u,q(u)\bigr) \,\mathrm{d}u, & \gamma_2=0, \end{array} %
\rright.
\nonumber
\\
& &\hspace*{-25pt}\qquad  { } + \mathrm{o}_P \bigl(k^{-1/2}\bigl(w_n(
\gamma_1)\vee w_n(\gamma_2)\bigr) \bigr).\nonumber
\end{eqnarray}
\end{corollary}

\begin{pf}
In view of the Lemmas \ref{lemma:Snstarapprox}--\ref{lemma:singleintapprox}, it suffices to define sequences
$\lambda_{n,i}$ and $\tau_{n,i}$, $i\in\{1,2\}$, such that the
conditions (i)--(iv) of Lemma~\ref{lemma:marginalapprox} and
\eqref{Q2.1} and \eqref{Q2.2} are fulfilled and
\begin{eqnarray*}
\frac{\lambda_{n,1}-x_l}{q^2(\lambda_{n,1}-)} + \frac{q(\tau_{n,1})-q(\infty)}{\tau_{n,1}^2} & =&\mathrm{o} \bigl(k^{-1/2}
w_n(\gamma_1) \bigr),
\\
\frac{\lambda_{n,2}-q(\infty)}{(q^\leftarrow(\lambda_{n,2}))^2} +\frac{q^\leftarrow(\tau_{n,2})-x_l}{\tau_{n,2}^2}& =&\mathrm{o} \bigl(k^{-1/2}
w_n(\gamma_2) \bigr).
\end{eqnarray*}
Note that we can check these conditions for $i=1$ and $i=2$
separately. We focus on the sequences $\lambda_{n,1}$ and
$\tau_{n,1}$, since the case $i=2$ can be treated analogously if
$x_l$ is replaced with $q(\infty)$ and $q$ with $q^\leftarrow$.
Again we distinguish three cases depending on the sign of
$\gamma_1$.

If $\gamma_1>0$, then $\tau_{n,1}$ must only satisfy
$(q(\tau_{n,1})-q(\infty))/ \tau_{n,1}^2 = \mathrm{o}(k^{-1/2}\log c_n)$,
which can easily be fulfilled by letting $\tau_{n,1}$ tend to
$\infty$ sufficiently fast.

The sequence $\lambda_{n,1}$ has to satisfy the conditions (i) and
(ii) of Lemma~\ref{lemma:marginalapprox}, \eqref{Q2.1} and
$(\lambda_{n,1}-x_l)/q^2(\lambda_{n,1})=\mathrm{o}(k^{-1/2}\log c_n)$. If
$x_l>0$, then $\lambda_{n,1}=x_l$ does the job, because condition
(i) of Lemma~\ref{lemma:marginalapprox} is implied by (S2).

If $x_l=0$ and $\gamma_1\le1$, then the integrability condition of
\eqref{Q2.1} is trivial. Moreover, $\lambda_{n,1}:= k^{-1/2}(\log
c_n)^{1/2}\to0$ obviously fulfills Lemma \ref{lemma:marginalapprox}(ii)
and
$(\lambda_{n,1}-x_l)/q^2(\lambda_{n,1})=\mathrm{O}(\lambda
_{n,1})=\mathrm{o}(k^{-1/2}\log
c_n)$. Condition \ref{lemma:marginalapprox}(i) follows from (S2)
and (S3), which implies $c_n\lambda_{n,1}\to\infty$.

Finally, if $x_l=0$ and $\gamma_1>1$, then $\lambda_{n,1}:=
(k^{-1/2}\log c_n)^{1/\gamma_1}$ fulfills Lemma \ref{lemma:marginalapprox}(ii),
Lem\-ma~\ref{lemma:marginalapprox}(i) follows from (S2) and (S3) as
above, and (Q2) implies
\[
\frac{\lambda_{n,1}-x_l}{q^2(\lambda_{n,1})} =\mathrm{O} \biggl(\frac{\lambda_{n,1}^{\gamma_1}}{|\log\lambda_{n,1}|^2} \biggr) =\mathrm{O}
\biggl(k^{-1/2}\frac{\log c_n}{|\log(k^{-1/2} \log c_n)|^2} \biggr) = \mathrm{o}
\bigl(k^{-1/2}\log c_n\bigr)
\]
by (S1). Furthermore, the integrability condition of \eqref{Q2.1} is
fulfilled, because (Q2) implies $(v/\log
v)^{2/(1-\gamma_1)}=\mathrm{O}(q^\leftarrow(v))$ as $v\to\infty$.

Next, we consider the case $-1/2<\gamma_1<0$, when the
integrability condition of \eqref{Q2.1} is trivial. If $x_l>0$, then
we can argue as above that $\lambda_{n,1}=x_l$ satisfies all
conditions on $\lambda_{n,1}$. If $x_l=0$, then define
$\lambda_{n,1}=c_n^{-1}\varphi_n$ for some $\varphi_n\to\infty$
sufficiently slowly, so that Lemma \ref{lemma:marginalapprox}(i) follows
from (S2). Further
$(\lambda_{n,1}-x_l)/q^2(\lambda_{n,1})=\mathrm{O}(c_n^{-1}\varphi_n)=\mathrm{o}(k^{-1/2}
c_n^{-\gamma_1})$ follows from assumption (S3).

The conditions on $\tau_{n,1}$ read as
$(q(\tau_{n,1})-q(\infty))/\tau_{n,1}^2=\mathrm{o}(k^{-1/2}c_n^{-\gamma_1})$
and $k^{-1/2}=\linebreak[4] \mathrm{o}((c_n\tau_{n,1})^{\gamma_1})$ in this case, which are
fulfilled by $\tau_{n,1}=k^{1/2} c_n^{\gamma_1}\to\infty$.

In the case $\gamma_1=0$ the integrability condition of
\eqref{Q2.1} is again trivial and $\lambda_{n,1}=x_l$ if $x_l>0$,
and $\lambda_{n,1}=c_n^{-1}\log c_n$ if $x_l=0$ does the job.
Moreover, it is easily checked that $\tau_{n,1}=k^{1/4}$ satisfies
$(q(\tau_{n,1})-q(\infty))/\tau_{n,1}^2=\mathrm{o}(k^{-1/2}\log^2 c_n)$ and
condition of Lemma \ref{lemma:marginalapprox}(iv).
\end{pf}

Observe that we have verified stronger conditions on $\lambda_{n,1}$
and $\tau_{n,1}$ than actually necessary, if
$w_n(\gamma_1)=\mathrm{o}(w_n(\gamma_2))$. A refined analysis would lead to
weaker, but more complex conditions on $q$ and $k$ that depend on
both the values of $\gamma_1$ and $\gamma_2$ at the same time. (Also
the proof would become more lengthy as one had to consider 9 cases
arising from different combinations of signs of $\gamma_1$ and
$\gamma_2$.) Moreover, note that for the above choice of
$\lambda_{n,i}$ one has
%
%e5.20 #&#
\begin{equation}
\label{eq:cnlambdandiv} c_n \lambda_{n,i}\to\infty,\qquad  i\in\{1,2\}
\end{equation}
and
%
%e5.21 #&#
\begin{equation}
\label{eq:lambdanbound} \lambda_{n,i}^{-\gamma_i}=\mathrm{O}
\bigl(k^{1/2}/\log c_n\bigr)\qquad  \mbox{if }
\gamma_i>0, i\in\{1,2\}.
\end{equation}

Now we use classical empirical process theory to establish a uniform
bound on $\nu_n(B) - E\nu_n(B)$ and thus on term $\mathit{II}$ in
decomposition \eqref{eq:esterror}.

%le5.6 #&#
\begin{lemma} \label{lemma:empprocess}
Under the conditions of Theorem~\ref{theo:main}, one has
\[
\nu_n(B) - E\nu_n(B)|_{B=(d_n/e_n)H_n(S_n^*)} =
\mathrm{o}_P \bigl(k^{-1/2}\bigl(w_n(
\gamma_1)\vee w_n(\gamma_2)\bigr) \bigr).
\]
\end{lemma}

\begin{pf}
Note that by \eqref{eq:tildeHrep} one has
\[
\frac{d_n}{e_n} H_n(x_1,x_2) = \bigl(
\tilde H_{\th_1,\chi_1,\xi_1}^{(n,1)}(x_1),\tilde
H_{\th_2,\chi_2,\xi_2}^{(n,2)}(x_2) \bigr)
\]
for $(x_1,x_2)\in[u_n^*,\infty)\times[v_n^*,\infty)$ with
$ \tilde H_{\th_i,\chi_i,\xi_i}^{(n,i)}$ defined by
\eqref{eq_tildeHnidef} and
\[
-\th_i=\chi_i=k^{1/2}(\hat
\gamma_i-\gamma_i), \qquad \xi_i =
k^{1/2} \biggl(\frac{\hat a_i(n/k)}{a_i(n/k)} -1 - \frac{\hat
b_i(n/k)-b_i(n/k)}{a_i(n/k)}\hat
\gamma_i \biggr).
\]
Since, according to condition (M3), these random variables are
stochastically bounded, it suffices to prove that for all $M>0$
\[
\sup_{\max(|\th_i|,|\chi_i|,|\xi_i|)\le M} \bigl|\nu_n \bigl(E^{(n)}_{(\th_i,\chi_i,\xi_i)_{i=1,2}}
\bigr)-E \nu_n \bigl(E^{(n)}_{(\th_i,\chi_i,\xi_i)_{i=1,2}} \bigr) \bigr| =
\mathrm{o}_P \bigl(k^{-1/2}\bigl(w_n(
\gamma_1)\vee w_n(\gamma_2)\bigr) \bigr),
\]
where
\[
E^{(n)}_{(\th_i,\chi_i,\xi_i)_{i=1,2}} := \bigl\{ \bigl(\tilde H_{\th_1,\chi_1,\xi_1}^{(n,1)}(x_1),
\tilde H_{\th_2,\chi_2,\xi_2}^{(n,2)}(x_2) \bigr)
\mid(x_1,x_2)\in S_n^* \bigr\}.
\]
Letting $\theta:=(\th_i,\chi_i,\xi_i)_{i=1,2}$ and
\[
Z_n(\theta) := \frac{k^{1/2}}{w_n(\gamma_1)\vee w_n(\gamma_2)} \bigl(\nu_n
\bigl(E_\theta ^{(n)}\bigr)-E\nu_n
\bigl(E_\theta^{(n)}\bigr) \bigr), \qquad \theta\in[-M,M]^6,
\]
we have to prove that $Z_n$ tends to 0 in probability uniformly. To
this end, we
establish asymptotic equicontinuity of $Z_n$, that is,
%
%e5.22 #&#
\begin{equation}
\label{eq:equicont} \lim_{\delta\downarrow0} \limsup_{n\to\infty} P
\Bigl\{ \sup_{\theta,\psi\in[-M,M]^6,
\|\theta-\psi\|_\infty\le\delta}\bigl|Z_n(\theta)-Z_n(
\psi)\bigr|>\eta \Bigr\}=0\qquad  \forall\eta>0
\end{equation}
and convergence in probability of $Z_n(\theta)$ for all $\theta\in
[-M,M]^6$ (see van der Vaart and Wellner \cite{vw00}, Theorem~1.5.7).

For the proof of asymptotic equicontinuity, it is crucial that the
functions $\tilde
H_{\th_i,\chi_i,\xi_i}^{(n,i)}(x_i)$\vspace*{2pt} are decreasing in all three
parameters for all $(x_1,x_2)\in
[u_n^*,\infty)\times[v_n^*,\infty)$. For $\xi_i$ resp. $\th_i$
this monotonicity is an
immediate consequence of the facts that $(c_n^{-\gamma}-1)/\gamma$
is negative and increasing\nonumber1.5 in $\gamma$ (for $c_n>1$) and that
$(1+\gamma t)^{1/\gamma}$ is increasing in $t$. Because
$c_n^{-\gamma}$ is a decreasing function of $\gamma$, the
monotonicity in $\chi_i$ follows from \eqref{eq:secordunif2},
\eqref{eq:cnlambdandiv} and condition (i) of Lemma~\ref
{lemma:marginalapprox}, which imply
\begin{eqnarray*}
\frac{U_i(d_nx)-b_i(n/k)}{a_i(n/k)} & = & \frac{(x_i
d_nk/n)^{\gamma_i}-1}{\gamma_i} + \mathrm{O} \bigl(A_i(n/k)
(x_i d_nk/n)^{\gamma_i+\rho_i+\eps} \bigr)
\\
& = & \frac{(x_i c_n d_n/ e_n)^{\gamma_i}-1}{\gamma_i} + \mathrm{o} \bigl((c_nx_i)^{\gamma_i}k^{-1/2}w_n(
\gamma_i) \bigr)
\\
&>& 0
\end{eqnarray*}
for sufficiently large $n$.

The monotonicity of $H^{(n,i)}_{\cdot,\cdot,\cdot}(x_i)$ implies
that the sets $ E^{(n)}_{(\th_i,\chi_i,\xi_i)_{i=1,2}}$ are
increasing in all parameters. Hence, for arbitrary $\theta,\psi\in
[-M,M]^6$
\[
\bigl|Z_n(\theta)-Z_n(\psi)\bigr|\le\frac{k^{1/2}}{w_n(\gamma_1)\vee
w_n(\gamma_2)} \bigl(
\nu_n\bigl(E_{\theta\vee\psi}^{(n)}\setminus
E_{\theta\wedge\psi}^{(n)}\bigr) +E\nu_n\bigl(E_{\theta\vee\psi}^{(n)}
\setminus E_{\theta\wedge\psi
}^{(n)}\bigr) \bigr),
\]
where $\theta\vee\psi$ resp. $\theta\wedge\psi$ denote the
coordinatewise maximum resp. minimum of $\theta$ and $\psi$.

To establish asymptotic equicontinuity of $Z_n$, we cover the parameter
space $[-M,M]^6$ with hypercubes $I_l:=\bigtimes_{i=1}^6
[l_i\delta,(l_i+1)\delta]$, $-\lceil M/\delta\rceil\le l_i\le
\lfloor M/\delta\rfloor$, for some small $\delta>0$ (depending on
the value $\eta$ in \eqref{eq:equicont}) to be specified later on.
For $\theta,\psi\in
[-M,M]^6$ with $\|\theta-\psi\|_\infty\le\delta$ and
$l(\theta):= (\lfloor\theta_i/\delta\rfloor )_{1\le i\le
6}$, one has $\|l(\theta)-l(\psi)\|\le1$ and thus
%
%e5.23 #&#
\begin{eqnarray}\label{eq:Zndiffapprox}
&&\bigl|Z_n(\theta)-Z_n(\psi)\bigr|
\nonumber
\\
&&\quad  \le \bigl|Z_n(\theta)-Z_n\bigl(l(\theta)\delta\bigr)\bigr| +
\bigl|Z_n(\psi)-Z_n\bigl(l(\psi)\delta\bigr)\bigr| +
\bigl|Z_n\bigl(l(\theta)\delta\bigr)-Z_n\bigl(l(\psi)\delta
\bigr)\bigr|
\nonumber
\\[-8pt]\\[-8pt]
&&\quad  \le 3 \max_{l\in\{-\lceil M/\delta\rceil,\ldots,
\lfloor M/\delta\rfloor\}^6} \sup_{t,u\in I_l}\bigl|Z_n(t)-Z_n(u)\bigr|\nonumber
\\
& &\quad \le 3 \frac{k^{1/2}}{w_n(\gamma_1)\vee w_n(\gamma_2)} \max_{l\in\{-\lceil M/\delta
\rceil,\ldots,
\lfloor M/\delta\rfloor\}^6} \bigl(
\nu_n\bigl(E_{(l+1)\delta}^{(n)}\setminus
E_{l\delta
}^{(n)}\bigr)
 +E \nu_n\bigl(E_{(l+1)\delta}^{(n)}\setminus
E_{l\delta}^{(n)}\bigr) \bigr),\nonumber
\end{eqnarray}
where $(l+1)\delta:= ((l_i+1)\delta)_{1\le i\le6}$. By
(D1), the expectation can be approximated as
follows:
%
%e5.24 #&#
\begin{eqnarray}
\label{eq:nuEnapprox} E \nu_n\bigl(E_{(l+1)\delta}^{(n)}
\setminus E_{l\delta}^{(n)}\bigr) &=& \frac{n}k P\bigl
\{T_n^\leftarrow(X,Y)\in E_{(l+1)\delta}^{(n)}\setminus
E_{l\delta}^{(n)}\bigr\}\nonumber \\[-8pt]\\[-8pt]
&=& \nu \bigl(E_{(l+1)\delta}^{(n)}
\setminus E_{l\delta}^{(n)} \bigr) + \mathrm{O}
\bigl(A_0(n/k)\bigr).\nonumber
\end{eqnarray}
To bound the right-hand side, first note that by similar calculations
as in the
proof of Lemma~\ref{lemma:marginalapprox}, one obtains
\begin{eqnarray*}
&&\tilde H_{\th_i,\chi_i,\xi_i}^{(n,i)}(x)
\\
&&\quad  =  \frac{d_n}{e_n} x \left(\vphantom{\begin{array} {l@{\qquad  }l} -k^{-1/2} \log
c_n \biggl(\displaystyle \frac{\chi_i}{\gamma_i} +\mathrm{o}_P(1)\biggr)+
\mathrm{O}_P\bigl(k^{-1/2}(xd_n/e_n)^{-\gamma_i}
\bigr), & \gamma_i>0,
\\
k^{-1/2} (d_nk/n)^{-\gamma_i} \bigl(\bigl(
\xi_i/\gamma_i+\th_i/
\gamma_i^2+\mathrm{o}_P(1)
\bigr)x^{-\gamma
_i}+\mathrm{o}_P(1) \bigr), &
\gamma_i<0,
\\
-k^{-1/2} \log^2 c_n \bigl(\chi_i+
\th_i/2+\mathrm{o}_P(1)\bigr)+ \mathrm{O}_P
\bigl(k^{-1/2} \log c_n\log x\bigr), & \gamma_i=0,
\end{array}} \right.1
\\
& &\left.\hphantom{\quad  =  \frac{d_n}{e_n} x \biggl(}{} + \lleft\{ %
\begin{array} {l@{\qquad  }l} -k^{-1/2} \log
c_n \biggl(\displaystyle \frac{\chi_i}{\gamma_i} +\mathrm{o}_P(1)\biggr)+
\mathrm{O}_P\bigl(k^{-1/2}(xd_n/e_n)^{-\gamma_i}
\bigr), & \gamma_i>0,
\\
k^{-1/2} (d_nk/n)^{-\gamma_i} \bigl(\bigl(
\xi_i/\gamma_i+\th_i/
\gamma_i^2+\mathrm{o}_P(1)
\bigr)x^{-\gamma
_i}+\mathrm{o}_P(1) \bigr), &
\gamma_i<0,
\\
-k^{-1/2} \log^2 c_n \bigl(\chi_i+
\th_i/2+\mathrm{o}_P(1)\bigr)+ \mathrm{O}_P
\bigl(k^{-1/2} \log c_n\log x\bigr), & \gamma_i=0
\end{array} %
\rright. \right)
\end{eqnarray*}
uniformly for $x\in[\lambda_{n,i},\tau_{n,i}]$. That means that
under the same conditions as in Lemma~\ref{lemma:marginalapprox}
one can prove an analogous approximation where $\Gamma_i$ is
replaced with $\chi_i$ if $\gamma_i>0$,
$\alpha_i/\gamma_i-\beta_i-\Gamma_i/\gamma_i^2$ is replaced with
$\xi_i/\gamma_i+\th_i/\gamma_i^2$ if $\gamma_i<0$, and $\Gamma_i$
is replaced with $2\chi_i+\th_i$ in the case $\gamma_i=0$. Hence,
we may also conclude a corresponding analog to Corollary~\ref{cor:termIVVapprox}, that is,
$\nu((e_n/d_n)E^{(n)}_{(\th_i,\chi_i,\xi_i)_{i=1,2}})-\nu(S)$
equals the right-hand side of \eqref{eq:nuHnapprox} with the
above substitutions. Because all integrals are finite, there
exists a constant $K>0$ such that for sufficiently large $n$
\[
\nu\bigl(E_{(l+1)\delta}^{(n)}\bigr)-\nu\bigl(E_{l\delta}^{(n)}
\bigr)\le\frac
{e_n}{d_n} K\delta k^{-1/2} \bigl(w_n(
\gamma_1)\vee w_n(\gamma_2)\bigr)
\]
uniformly for all $l\in\{-\lceil M/\delta\rceil,\ldots,
\lfloor M/\delta\rfloor\}^6$. A combination with
\eqref{eq:nuEnapprox}, $e_n\asymp d_n$ and condition (S2) shows that
to each $\eta>0$ there exists
$\delta>0$ such that for sufficiently large~$n$
%
%e5.25 #&#
\begin{equation}
\label{eq:Enunbound} E \nu_n\bigl(E_{(l+1)\delta}^{(n)}
\setminus E_{l\delta}^{(n)}\bigr) \le\frac\eta{12} k^{-1/2}
\bigl(w_n(\gamma_1)\vee w_n(
\gamma_2)\bigr).
\end{equation}
In view of \eqref{eq:Zndiffapprox}, we obtain\vspace*{-2pt}
\begin{eqnarray*}
&&P \Bigl\{ \sup_{\theta,\psi\in[-M,M]^6,
\|\theta-\psi\|_\infty\le\delta}\bigl|Z_n(\theta)-Z_n(
\psi)\bigr|>\eta \Bigr\}
\\
&&\quad  \le P \biggl\{\max_{l\in\{-\lceil M/\delta\rceil,\ldots,
\lfloor M/\delta\rfloor\}^6} \bigl(\nu_n
\bigl(E_{(l+1)\delta}^{(n)}\setminus E_{l\delta}^{(n)}
\bigr)+E \nu _n\bigl(E_{(l+1)\delta}^{(n)}\setminus
E_{l\delta}^{(n)}\bigr) \bigr)
\\
& &\hphantom{\quad  \le P \biggl\{} > \frac\eta3 k^{-1/2} \bigl(w_n(\gamma_1)
\vee w_n(\gamma_2)\bigr) \biggr\}
\\
&&\quad  \le \sum_{l\in\{-\lceil M/\delta\rceil,\ldots,
\lfloor M/\delta\rfloor\}^6} P \biggl\{ \bigl|\nu_n
\bigl(E_{(l+1)\delta
}^{(n)}\setminus E_{l\delta}^{(n)}
\bigr)- E \nu_n\bigl(E_{(l+1)\delta}^{(n)}\setminus
E_{l\delta}^{(n)}\bigr) \bigr|
\\
& &\hphantom{\quad  \le \sum_{l\in\{-\lceil M/\delta\rceil\ldots
\lfloor M/\delta\rfloor\}^6} P \biggl\{} > \frac\eta6 k^{-1/2} \bigl(w_n(\gamma_1)
\vee w_n(\gamma_2)\bigr) \biggr\}.
\end{eqnarray*}
Therefore the asserted asymptotic equicontinuity \eqref{eq:equicont} follows
from \eqref{eq:Enunbound} and Chebyshev's inequality applied to the
binomial random
variables $k \nu_n(E_{(l+1)\delta}^{(n)}\setminus
E_{l\delta}^{(n)})$:\vspace*{-2pt}
\begin{eqnarray*}
&&P \biggl\{ \bigl|\nu_n\bigl(E_{(l+1)\delta}^{(n)}\setminus
E_{l\delta}^{(n)}\bigr)- E \nu_n\bigl(E_{(l+1)\delta}^{(n)}
\setminus E_{l\delta}^{(n)}\bigr) \bigr| > \frac\eta6 k^{-1/2}
\bigl(w_n(\gamma_1)\vee w_n(
\gamma_2)\bigr) \biggr\}
\\
&&\quad  \le \frac{k E \nu_n(E_{(l+1)\delta}^{(n)}\setminus
E_{l\delta}^{(n)})}{(\eta/6)^2
k(w_n(\gamma_1)\vee w_n(\gamma_2))^2} \to0
\end{eqnarray*}
uniformly for all $l\in\{-\lceil M/\delta\rceil,\ldots,
\lfloor M/\delta\rfloor\}^6$.

It remains to prove that $Z_n(\theta)\to0$ in probability for all
$\theta\in[-M,M]^6$. This, however, follows similarly by Chebyshev's
inequality, (D1) and the aforementioned analog to Corollary~\ref
{cor:termIVVapprox}:\vspace*{-2pt}
\begin{eqnarray*}
P\bigl\{\bigl|Z_n(\th)\bigr|>\eta\bigr\}
&  =&  P \bigl\{k \bigl|\nu_n\bigl(E_{\theta}^{(n)}\bigr)-
E \nu_n\bigl(E_{\theta}^{(n)}\bigr) \bigr| > \eta
k^{1/2} \bigl(w_n(\gamma_1)\vee
w_n(\gamma_2)\bigr) \bigr\}
\\
&  \le& \frac{nP\{T_n^\leftarrow(X,Y)\in E_{\theta}^{(n)}\} }{\eta^2
k(w_n(\gamma_1)\vee w_n(\gamma_2))^2}
\\
&  =&  \frac{\nu(E_{\theta}^{(n)})+\mathrm{O}(A_0(n/k))}{\eta^2
(w_n(\gamma_1)\vee w_n(\gamma_2))^2}
\\
&  =&  \frac{\nu(S)+\mathrm{o}(1)}{\eta^2
(w_n(\gamma_1)\vee w_n(\gamma_2))^2}
\\
&  \to& 0.
\end{eqnarray*}
\upqed
\end{pf}

%re5.7 #&#
\begin{remark}
Two remarks on this proof are in place. At first glance it seems
peculiar that in the definition of $\tilde
H_{\th_i,\chi_i,\xi_i}^{(n,i)}$ both parameters $-\th_i$ and $\chi
_i$ take over the role of
$k^{1/2}(\hat\gamma_i-\gamma_i)$ in the definition of $\tilde H$.
This, however, is necessary to ensure the crucial monotonicity
property of $\tilde
H_{\th_i,\chi_i,\xi_i}^{(n,i)}$ in the case $\gamma_i>0$.\vspace*{2pt}

Second, we used the (slightly old-fashioned) classical approach
to establish asymptotic equicontinuity instead of the often more elegant
approach via bracketing numbers (see van der Vaart and Wellner \cite{vw00},
Theorem~2.11.9), because the same approximation error of order
$\mathrm{O}(A_0(n/k))$ in (D1)
always enters the upper bound on $E \nu_n(E_{(l+1)\delta
}^{(n)}\setminus
E_{l\delta}^{(n)})$, thus impeding the calculation of bracketing
numbers for radii of smaller order.
\end{remark}

Next, we show that the terms $I$ and $\mathit{III}$ in decomposition
\eqref{eq:esterror} are negligible.

%le5.8 #&#
\begin{lemma} \label{lemma:truncInegl}
If the conditions of Theorem~\ref{theo:main} are
fulfilled, then
%
%e5.26 #&#
%e5.27 #&#
\begin{eqnarray}
\label{eq:termIbound} \hat p_n - \frac{1}{e_n} \nu_n
\biggl(\frac{d_n}{e_n} H_n\bigl(S_n^*\bigr) \biggr) & =
& \mathrm{o}_P \bigl( d_n^{-1}k^{-1/2}
\bigl(w_n(\gamma_1)\vee w_n(
\gamma_2)\bigr) \bigr),
\\
\label{eq:termIIIbound} \frac{1}{e_n} \bigl( E\nu_n(B)-\nu(B)
\bigr)|_{B=(d_n/e_n)H_n(S_n^*)} & = & \mathrm{o}_P \bigl(
d_n^{-1}k^{-1/2} \bigl(w_n(
\gamma_1)\vee w_n(\gamma_2)\bigr) \bigr).
\end{eqnarray}
\end{lemma}

\begin{pf}
As $\hat p_n= \nu_n( (d_n/e_n)H_n(S) )/e_n$, the left-hand
side of \eqref{eq:termIbound} is
non-negative with expectation
\begin{eqnarray*}
&&\frac{n}{ke_n} P \biggl\{T_n^\leftarrow(X,Y)\in
\frac{d_n}{e_n} H_n\bigl(S\setminus S_n^*\bigr) \biggr
\}
\\
&&\quad  \le \frac{n}{ke_n} P \biggl\{T_n^\leftarrow(X,Y)\in
\frac{d_n}{e_n} H_n \bigl( \bigl(0,u_n^*\bigr)\times\bigl[q
\bigl(u_n^*-\bigr),\infty\bigr) \cup \bigl[q^\leftarrow
\bigl(v_n^*\bigr),\infty\bigr)\times \bigl[q(\infty),v_n^*\bigr) \bigl) \biggr
\}
\\
&&\quad  =  \frac{1}{d_n} \bigl( \nu \bigl(H_n \bigl(
\bigl(0,u_n^*\bigr)\times\bigl[q\bigl(u_n^*-\bigr),\infty\bigr)
\cup \bigl[q^\leftarrow\bigl(v_n^*\bigr),\infty\bigr)\times \bigl[q(
\infty),v_n^*\bigr) \bigr)\bigr)
\\
& &\quad\hphantom{\frac{1}{d_n} \bigl( \nu\bigl(}{} + \mathrm{o} \bigl( k^{-1/2} \bigl(w_n(
\gamma_1)\vee w_n(\gamma_2)\bigr) \bigr) \bigr),
\end{eqnarray*}
%
% & = & \frac1{e_n}\Big(
% \nu\Big(\frac{d_n}{e_n} H_n(S\setminus S_n^*)\Big)+
% \mathrm{O}\big(A_0(n/k)\big)\Big)\\
% & = & \frac1{d_n}\Big(
% \nu\big(H_n(S\setminus S_n^*)\big)+
% \mathrm{o}\big( k^{-1/2}
% (w_n(\gamma_1)\vee w_n(\gamma_2))\big)\Big),
where we have used (D1) and (S2). Now assertion \eqref{eq:termIbound}
follows from
Lemma~\ref{lemma:Snstarapprox} and the proof of Corollary~\ref
{cor:termIVVapprox}.

Likewise, by conditions (D1), (S2) and $d_n\asymp e_n$, the
left-hand side of \eqref{eq:termIIIbound} equals
\begin{eqnarray*}
&& \frac{1}{e_n} \biggl(\frac{n}k P\bigl\{ T_n^\leftarrow(X,Y)
\in B\bigr\}-\nu(B) \biggr) \biggl|_{B=(d_n/e_n)H_n(S_n^*)}
\\
&&\quad   =  \mathrm{O}_P \bigl(e_n^{-1}
A_0(n/k) \bigr)
\\
&&\quad   =  \mathrm{o}_P \bigl( d_n^{-1}k^{-1/2}
\bigl(w_n(\gamma_1)\vee w_n(
\gamma_2)\bigr) \bigr).
\end{eqnarray*}
\upqed
\end{pf}

Finally, we derive a bound on term $\mathit{VI}$ in decomposition
\eqref{eq:esterror}.

%le5.9 #&#
\begin{lemma} \label{lemma:termVI}
Under the assumptions of Theorem~\ref{theo:main} one has
\[
\nu(d_nS)-p_n = \mathrm{o} \bigl(d_n^{-1}k^{-1/2}
\bigl(w_n(\gamma_1)\vee w_n(
\gamma_2)\bigr) \bigr).
\]
\end{lemma}

\begin{pf}
With $\lambda_{n,i},\tau_{n,i}$ as in Lemma~\ref{lemma:marginalapprox}, we define for
$x\in[\lambda_{n,i},\tau_{n,i}]$
\[
H_{n,i}^*(x) := \biggl(1+\gamma_i\frac{U_i(d_nx)-b_i(d_n)}{a_i(d_n)}
\biggr)^{1/\gamma_i}.
\]
According to de Haan and Ferreira \cite{hf06}, Theorems 2.3.6 and 2.3.7
one can choose $a_i(t)$ as a multiple of $t^{\gamma_i}$ and
$b_i(t)=U_i(t)+\mathrm{O}(a_i(t)A_i(t))$. Thus, for $\Delta_1(x)$ defined in
the proof of Lemma~\ref{lemma:marginalapprox}
\begin{eqnarray*}
&&\frac{U_i(d_nx)-b_i(d_n)}{a_i(d_n)}
\\
&&\quad  =  \frac{a_i(n/k)}{a_i(d_n)} \biggl(\frac{U_i(d_nx)-b_i(n/k)}{a_i(n/k)}-\frac
{b_i(d_n)-b_i(n/k)}{a_i(n/k)} \biggr)
\\
&&\quad  =  \biggl(\frac{n}{kd_n} \biggr)^{\gamma_i} \biggl(
\frac{(xd_nk/n)^{\gamma_i}-1}{\gamma
_i}+\Delta_1(x)+ \frac{(d_nk/n)^{\gamma_i}-1}{\gamma_i}+
\Delta_1(1) \biggr)+\mathrm{O}\bigl(A_i(d_n)
\bigr)
\\
&&\quad  =  \frac{x^{\gamma_i}-1}{\gamma_i} + \mathrm{O} \biggl(A_i(n/k) \biggl(
\frac{d_nk}n \biggr)^{\rho_i+\eps}\bigl(x^{\gamma
_i+\rho_i+\eps}+1\bigr) \biggr) +
\mathrm{o} \biggl(A_i(n/k) \biggl(\frac{d_nk}n
\biggr)^{\rho_i+\eps} \biggr),
\end{eqnarray*}
where in the last step we have used \eqref{eq:Delta1}, \eqref
{eq:cnlambdandiv} and the Potter bound for the
regularly varying function $A_0$ (de Haan and Ferreira \cite{hf06},
Proposition B.1.9 5.). We conclude that
\begin{eqnarray*}
1+\gamma_i \frac{U_i(d_nx)-b_i(d_n)}{a_i(d_n)} &=& x^{\gamma_i} \biggl(1+
\mathrm{O} \biggl(A_i(n/k) \biggl(\frac{xd_n k}n
\biggr)^{\rho
_i+\eps} \biggr) \\
&&\hphantom{x^{\gamma_i} \biggl(}{}+ \mathrm{O} \biggl(A_i(n/k) \biggl(
\frac{d_n
k}n \biggr)^{\rho_i+\eps}x^{-\gamma_i} \biggr) \biggr).
\end{eqnarray*}
Check that the first remainder term is of smaller order than
$k^{-1/2}w_n(\gamma_i)$
by condition (i) of Lemma~\ref{lemma:marginalapprox}. Moreover,
for $\gamma_i>0$, \eqref{eq:lambdanbound} and again condition (i) of
Lemma~\ref{lemma:marginalapprox} imply
\[
A_i(n/k) \biggl(\frac{d_n k}n \biggr)^{\rho_i+\eps}x^{-\gamma_i}
= \mathrm{O} \biggl(A_i(n/k) \biggl(\frac{d_n k}n
\biggr)^{\rho_i+\eps} k^{1/2}/\log c_n \biggr) \to0,
\]
while for $\gamma_i<0$ this convergence follows from the conditions
(i) and (iii) of Lemma~\ref{lemma:marginalapprox}, and for $\gamma_i$ it is obvious from
condition (i).

This shows that $H_{n,i}^*(x)$ is indeed well defined with
\[
H_{n,i}^*(x) = x \left(\vphantom{\begin{array} {l@{\qquad  }l}
\mathrm{o}\bigl(k^{-1/2} \log c_n\bigr)+\mathrm{O}
\bigl(A_i(n/k) (d_nk/n)^{\rho_i+\eps}x^{-\gamma_i}
\bigr), & \gamma_i>0,
\\\noalign{\vspace*{2pt}}
\mathrm{o} \bigl(k^{-1/2} (d_nk/n)^{-\gamma_i}
\bigl(1+x^{-\gamma_i}\bigr) \bigr), &\gamma_i<0,
\\\noalign{\vspace*{2pt}}
\mathrm{o}\bigl(k^{-1/2} \log^2 c_n\bigr), &
\gamma_i=0 \end{array}} 1+ \lleft\{ %
\begin{array} {l@{\qquad  }l}
\mathrm{o}\bigl(k^{-1/2} \log c_n\bigr)+\mathrm{O}
\bigl(A_i(n/k) (d_nk/n)^{\rho_i+\eps}x^{-\gamma_i}
\bigr), & \gamma_i>0,
\\\noalign{\vspace*{2pt}}
\mathrm{o} \bigl(k^{-1/2} (d_nk/n)^{-\gamma_i}
\bigl(1+x^{-\gamma_i}\bigr) \bigr), &\gamma_i<0,
\\\noalign{\vspace*{2pt}}
\mathrm{o}\bigl(k^{-1/2} \log^2 c_n\bigr), &
\gamma_i=0 \end{array} %
\rright.\right)
\]
uniformly for $x\in[\lambda_{n,i},\tau_{n,i}]$. Notice that this
representation
is of similar type as the approximation derived in Lemma~\ref{lemma:marginalapprox} with all leading terms equal to 0
(though in the case $\gamma_i>0$ the second remainder term has a
slightly different form). Therefore, we may proceed as before to
conclude
\begin{eqnarray*}
&&\nu\bigl(H_n^*\bigl(S_n^*\bigr)\bigr)-\nu(S)
\\
& &\quad  = \mathrm{o} \bigl(k^{-1/2}\bigl(w_n(
\gamma_1)\vee w_n(\gamma_2)\bigr) \bigr) +
\sum_{i=1}^2 \mathrm{O}
\bigl(A_i(n/k) (d_nk/n)^{\rho_i+\eps}
\bigr)1_{\{\gamma_i>0\}}
\\
&&\quad  =  \mathrm{o} \bigl(k^{-1/2}\bigl(w_n(
\gamma_1)\vee w_n(\gamma_2)\bigr) \bigr),
\end{eqnarray*}
where the last equality follows from Lemma~\ref{lemma:marginalapprox}(i) (cf. Corollary~\ref{cor:termIVVapprox}).

To complete the proof, we must show that
\[
p_n-\nu \bigl( d_nH_n^*\bigl(S_n^*
\bigr) \bigr) = \mathrm{o} \bigl(d_n^{-1}k^{-1/2}
\bigl(w_n(\gamma_1)\vee w_n(
\gamma_2)\bigr) \bigr).
\]
This, however, follows from assumption (D1) (with $t=d_n$) in a
similar way as \eqref{eq:termIbound}.
\end{pf}

\begin{pf*}{Proof of Theorem~\ref{theo:main}}
The assertion is a direct consequence of \eqref{eq:esterror},
Corollary~\ref{cor:termIVVapprox} and
of the Lemmas \ref{lemma:truncInegl}, \ref{lemma:empprocess} and
\ref{lemma:termVI}.
\end{pf*}

\begin{pf*}{Proof of Corollary~\ref{cor:asvarest}}
First note that, similarly as for $\hat p_n$, one obtains the
representation
$\hat\nu_n(\hat S_{n,2}^+)= \nu_n (\frac{d_n}{e_n} H_n^+(S)
)$ with
$H_n^+(x,y):= (H_{n,1}(x),H_{n,2}^+(y) )$,
\[
H_{n,2}^+(y) := \frac{e_n}{d_n} T_n^\leftarrow\circ
\hat T_n\circ\bigl(\hat T_n^{(c^+)}
\bigr)^\leftarrow\circ U(d_n y)
\]
and $c^+:= c_n^+ := (1+\ell_n)n/(ke_n)$. Thus, Lemma~\ref{lemma:marginalapprox} (with $e_n$ replaced by
$e_n/(1+\ell_n)$) yields the approximation
\begin{eqnarray*}
&&H_{n,2}^+(y) \\
&&\quad = (1+\ell_n) y \left(\vphantom{\begin{array} {l@{\qquad  }l} -k^{-1/2} \log
c_n \bigl(\Gamma_2/\gamma_2 +
\mathrm{o}_P(1)\bigr)+\mathrm{O}_P
\bigl(k^{-1/2}(yd_n/e_n)^{-\gamma_2}\bigr), &
\gamma_2>0,
\\
k^{-1/2} (d_nk/n)^{-\gamma_2} \bigl(\bigl(
\alpha_2/\gamma_2-\beta_2-
\Gamma_2/\gamma _2^2+\mathrm{o}_P(1)
\bigr)y^{-\gamma_2}+\mathrm{o}_P(1) \bigr), &
\gamma_2<0,
\\
-k^{-1/2} \log^2 c_n \bigl(\Gamma_2/2+
\mathrm{o}_P(1)\bigr)+ \mathrm{O}_P
\bigl(k^{-1/2} \log c_n\log y\bigr), & \gamma_2=0.
\end{array}}\right. 1
\\
& &\left.\hphantom{\quad = (1+\ell_n) y \bigl(}{} + \lleft\{ %
\begin{array} {l@{\qquad  }l} -k^{-1/2} \log
c_n \bigl(\Gamma_2/\gamma_2 +
\mathrm{o}_P(1)\bigr)+\mathrm{O}_P
\bigl(k^{-1/2}(yd_n/e_n)^{-\gamma_2}\bigr), &
\gamma_2>0,
\\
k^{-1/2} (d_nk/n)^{-\gamma_2} \\
\quad {}\times\bigl(\bigl(
\alpha_2/\gamma_2-\beta_2-
\Gamma_2/\gamma _2^2+\mathrm{o}_P(1)
\bigr)y^{-\gamma_2}+\mathrm{o}_P(1) \bigr), &
\gamma_2<0,
\\\noalign{\vspace*{2pt}}
-k^{-1/2} \log^2 c_n \bigl(\Gamma_2/2+
\mathrm{o}_P(1)\bigr)+ \mathrm{O}_P
\bigl(k^{-1/2} \log c_n\log y\bigr), & \gamma_2=0
\end{array} %
\rright. \hspace*{-3pt}\right).
\end{eqnarray*}
Now the very same arguments as used in the analysis of $\hat p_n$
show that
\[
\hat\nu_n\bigl(\hat S_{n,2}^+\bigr) = \nu
\bigl(S_{n,2}^+\bigr)+ \mathrm{O}_P \bigl(k^{-1/2}
\bigl(w_n(\gamma_1)\vee w_n(
\gamma_2)\bigr) \bigr).
\]
\goodbreak
Together with an analogous approximation for $\hat\nu_n(\hat
S_{n,2}^-)$ and our assumption on
$\ell_n$, we may conclude that
\begin{eqnarray*}
\frac{d_n}{e_n} \hat I_{n,2} & = & \frac{d_n}{e_n}
\frac{\nu
(S_{n,2}^-)-\nu(S_{n,2}^+)}{2\ell_n}+\mathrm{o}_P(1)
\\
& = & \int_{x_l}^\infty(2\ell_n)^{-1}
\int_{(1-\ell_n)q(u)}^{(1+\ell_n)q(u)}\eta(u,v) \,\mathrm{d}v \,\mathrm{d}u
\\
& \to& \int_{x_l}^\infty q(u)\eta\bigl(u,q(u)\bigr)
\,\mathrm{d}u.
\end{eqnarray*}
In the last step we have used the fact that, on the range of
integration, $\eta(u,v)$ is
continuous and bounded by a multiple of $u^{-3}\vee(q(u))^{-3}$
(cf. \eqref{eq:etanurel}),
so that the integrand of the outer integral can easily be bounded by
an integrable function and convergence follows by the dominated
convergence theorem.
\end{pf*}

% zodis "Acknowledgments" paliekamas pagal autoriu
\section*{Acknowledgements}
The research was partially supported by
grant FCT/PTDC/MAT/ 112770/2009 of the Portuguese National Foundation
for Science and
Technology FCT. We thank two anonymous referees and the associate
editor for their thoughtful comments which helped to improve the
presentation significantly.

%suskaldyti doi

% imsref loaded by jurgita.kaciuliene, 2014-03-20 13:41:52

\printhistory

\end{document}